\documentclass[format=sigconf, review=false]{acmart}
\usepackage{acm-arxiv}

\usepackage{balance} 
\usepackage{tikz}

\DeclareMathOperator*{\argmax}{arg\,max}

\newtheorem{rthm}{Theorem}
\newenvironment{reptheorem}[1]
  {\renewcommand\therthm{\ref*{#1}}\rthm}
  {\endrthm}

\definecolor{paperBlue}{HTML}{CCE1F4}
\definecolor{paperGreen}{HTML}{CBEBDD}
\definecolor{paperRed}{HTML}{F4D5E1}
\definecolor{paperYellow}{HTML}{F6E7C1}
\newcommand{\vb}[2]{{\setlength{\fboxsep}{1pt}\colorbox{#1}{\vphantom{A}\textsf{#2}}}}

\newcommand{\cA}{\vb{paperBlue}{A}}
\newcommand{\cB}{\vb{paperGreen}{B}}
\newcommand{\cC}{\vb{paperRed}{C}}

\AtEndPreamble{%
  \theoremstyle{acmdefinition}%
}

\title[Trembling toward Duverger]{Trembling toward Duverger: Convergence and Cycles in Synchronous Plurality Dynamics}

\author{Maximilian Weigel}
\affiliation{
  \institution{ETH Z\"urich}
  \city{Z\"urich}
  \country{Switzerland}}
\email{mweigel@ethz.ch}

\author{Damien Berriaud}
\affiliation{
  \institution{ETH Z\"urich}
  \city{Z\"urich}
  \country{Switzerland}}
\email{dberriaud@ethz.ch}

\author{Roger Wattenhofer}
\affiliation{
  \institution{ETH Z\"urich}
  \city{Z\"urich}
  \country{Switzerland}}
\email{wattenhofer@ethz.ch}

\begin{abstract}
Duverger's law states that plurality elections tend to concentrate support on two candidates. Pivotal-voter models typically derive this from aggregate uncertainty about votes, so that two-candidate competition emerges in the limit of large electorates. We instead propose a model for a fixed, finite population. A poll reveals the current support counts, and each voter best responds under trembling-hand beliefs, assuming that every other voter may deviate with a small probability.
Starting from truthful votes, we show that a single update typically suffices for the dynamics to converge to a stable profile in which every voter supports her preferred candidate among the two poll leaders. For any number of candidates and any full-support tremble belief, a four-vote lead of the runner-up over all lower candidates suffices; for three candidates, a lead of two suffices. As a consequence, under random preferences, Duverger's law emerges after one update with probability $1-\mathcal O(N^{-1/2})$.
While these results may suggest that the dynamics always settle, convergence can fail. Mild period-two cycles already arise, with a few voters swapping ballots around an otherwise stable profile. We show that substantially richer behavior is possible, even with three candidates.
Period-three cycles arise with six voters changing their ballots, and no such cycle can involve fewer than four. Cycles can also alternate the plurality winner.
Moreover, larger electorates support ever longer cycles: for every odd $\ell\geq3$, truthful voting can lead to a cycle whose vote counts repeat every $\ell$ rounds and whose full ballot profile repeats every $2\ell$ rounds.
\end{abstract}



\newcommand{\BibTeX}{\rm B\kern-.05em{\sc i\kern-.025em b}\kern-.08em\TeX}

\begin{document}
\maketitle

\section{Introduction}

A few weeks before an election, a poll shows two front-runners, a centrist narrowly in third place, and a fourth candidate far behind. Supporters of the fourth candidate do the sensible thing: they rally behind the centrist, their second choice. Meanwhile, some of the centrist's wavering supporters, who prefer her only slightly to one of the front-runners, conclude that she cannot win and defect. Both groups act on the same poll, at the same time. The defectors are then surprised by the next poll: the centrist, their favorite all along, is now in the lead. Everyone revises again, and the poll moves again. Do such revisions ever settle, or can polls and ballots keep chasing each other forever?

Each individual move in this story is the familiar logic behind Duverger's law: plurality rule tends to concentrate support on two candidates, because voters desert those who cannot win~\cite{duverger1954political,cox1997making}. 
Formal accounts of this coordination usually avoid the issue by relying on large electorates, deriving pivot probabilities from aggregate uncertainty about votes and obtaining two-candidate competition in the limit~\cite{palfrey1989duverger,myerson1993voting}. We instead fix a finite electorate, let all voters respond to the same poll simultaneously, and follow the resulting ballot profiles exactly.

The difficulty in a finite electorate is that most voters are not pivotal at the observed counts, so all ballots lead to the same outcome. Following Obraztsova et al.\ \cite{obraztsova2016trembling}, we resolve this indifference with trembling-hand beliefs: each voter treats the other ballots as subject to small independent deviations and chooses the ballot with the highest expected utility as the deviation probability tends to zero. These deviations are beliefs, not realized errors, and they select a deterministic response in every finite election. Applying this response to all voters simultaneously after each poll defines the dynamics we study, which Obraztsova et al.\ leave as an open direction.

\paragraph{Contributions.}
Our first contribution, in Section~\ref{sec:fixation}, identifies when one simultaneous revision ends the dynamics. If the runner-up leads every lower candidate by at least four votes, every voter switches to her preferred poll leader and no voter wants to move again: the profile becomes a sincere runoff between the leaders. This holds for any number of candidates and any full-support tremble belief; with three candidates a lead of two suffices, and parity decides the remaining cases. Failure thus requires a near-tie in the initial poll, which has probability $\mathcal O(N^{-1/2})$ under random preferences.

Our second contribution, in Section~\ref{sec:cycles}, shows that convergence can nevertheless fail, already with three candidates, and in escalating ways. Beyond the period-two swap in an election of \citet{obraztsova2016trembling},
period-three cycles with three candidates arise with six active voters, and never with fewer than four. Cycles can alternate the plurality winner, as in the story above. Finally, for every odd $\ell\geq3$, an election with $3(\ell+4)$ voters started from truthful votes enters a cycle whose poll has period $\ell$ and whose ballot profile has period $2\ell$.

\section{Related work}\label{sec:related_work}

\paragraph{Duverger's law and pivotal voting.}
Duverger's law asserts that plurality rule concentrates support on two candidates \cite{duverger1954political}, which Cox \cite{cox1997making} explains as strategic coordination. Formal derivations usually rely on large electorates, where pivot probabilities stem from aggregate uncertainty \cite{palfrey1989duverger,myerson1993voting}. Fey \cite{fey1997stability} shows that non-Duvergerian equilibria are unstable under poll-driven belief adjustment, and Messner and Polborn \cite{messner2011miscounts} obtain exactly two vote-getters in equilibria robust to rare miscounts. Forand and Maheshri \cite{forand_dynamic_2015} give a dynamic account based on party entry and exit rather than voter coordination. These works study equilibria and their stability; we follow individual ballots exactly in a fixed finite electorate.

\paragraph{Non-pivotal voters and trembling hands.}
In a finite election most voters are not pivotal, so Nash equilibrium predicts little. Refinements break this indifference through truth bias \cite{obraztsova_plurality_2013}, abstention \cite{desmedt_equilibria_2010}, or local dominance \cite{meir2014local,meir2015uncertainty}. Obraztsova et al.\ \cite{obraztsova2016trembling} adapt trembling-hand perfection \cite{selten1975perfectness} to plurality and characterize the resulting equilibria. Their Proposition~2 is the static precursor of our separation theorem; their Example~4 admits no equilibrium, and they list trembling-hand best-response dynamics as an open direction. We analyze exactly these dynamics under synchronous updates. Since our update is a deterministic map on a finite set (Lemma~\ref{lem:small-p}), their Example~4 must cycle; it does so with period two, as two voters swap ballots around an otherwise stable profile.

\paragraph{Iterative and poll-based voting.}
Under plurality, sequential best responses converge given suitable tie-breaking \cite{meir2010convergence,meir2017acyclic}, while other rules may not \cite{lev_convergence_2012}; see \cite{meir2018strategic} for a survey. Further work asks which equilibria are reached and how good they are \cite{thompson_empirical_2013,rabinovich_analysis_2015,branzei_how_2013,kavner_strategic_2021,kavner2024average}. Poll-based models let voters respond to published scores \cite{reijngoud2012polls,wilczynski2019poll}, which also shape human votes \cite{fairstein2019polls}. Closest to us, Vasselai \cite{vasselai2022icv} lets all voters respond synchronously to pivot probabilities estimated from current vote shares, proving convergence for large electorates and observing it in simulations of small ones. We derive responses from vanishing tremble beliefs instead, and exhibit cycles with only three candidates.

\paragraph{Cycles under simultaneous updates.}
Simultaneous revisions can cycle where sequential ones converge \cite{meir2017acyclic}. Kloeckner \cite{kloeckner2022cycles} proves that synchronous cycles are robust to small perturbations and exhibits bad cycles under approval voting. Goldberg et al.\ \cite{goldberg2026complexity} show that detecting convergence or cycles is NP-complete for iterative plurality, and build truthful period-$m$ cycles with two voters and $m$ candidates under top-bottom updates. We instead fix three candidates and obtain periods growing with the electorate, with count period $\ell$ and profile period $2\ell$.

\section{Model: Polls, Trembles, and Synchronous Update}\label{sec:model}

Let $M$ contain $m\geq2$ candidates and let $[N]=\{1,\ldots,N\}$ index the voters. Voter $i$ has fixed utilities $u_i:M\to\mathbb R$ with distinct values. At round $t$, the complete ballot profile is $S^t=(s_i^t)_{i\in[N]}$, and a poll publishes the counts $n_c^t=|\{i:s_i^t=c\}|$. Write $\mathbf e_c$ for the unit vector of candidate $c$ and $W(n)=\argmax_{c\in M}n_c$ for the plurality winner set. A nonempty tied set is valued by its uniform lottery $u_i(W)=\frac{1}{|W|}\sum_{c\in W}u_i(c).$
This convention is used only in the expected-utility calculation; no final winner is drawn.

\paragraph{Trembling-hand responses.}
After the poll, each voter reconsiders her ballot. She first removes her own ballot from the poll, giving the private counts $q_i^t=n^t-\mathbf e_{s_i^t}$; hence identical voters with different current ballots may respond differently. Taking the other ballots as certain would usually leave her indifferent, since a single ballot rarely changes the winner. Following Obraztsova et al.\ \cite{obraztsova2016trembling}, she instead believes that each other ballot independently stays with probability $1-p$ and otherwise switches uniformly to another candidate, so a specified directed tremble has probability $p/(m-1)$. These trembles are beliefs, not realized errors. If $\widetilde q$ denotes the resulting random private counts, the expected payoff of ballot $c$ is $U_i^p(c\mid q)=\mathbb E\bigl[u_i(W(\widetilde q+\mathbf e_c))\bigr]$.

For a fixed private count $q$, every pairwise payoff difference is a polynomial in $p$, and its first nonzero coefficient determines its sign for all sufficiently small $p>0$. To compute these coefficients, we call a realization $v$ of $\widetilde q$ a \emph{terminal} private vector. A \emph{flow} specifies how many of the other ballots tremble from each candidate to each other candidate; for instance, $k(A\mapsto B)$ denotes $k$ ballots trembling from $A$ to $B$. Many individual tremble assignments yield the same flow, and their number is its \emph{multiplicity}. In words, a voter compares ballots on the most likely flows under which her ballot matters. We write $\operatorname{BR}_i(q)$ for the resulting best response.

\paragraph{Synchronous update.}
All voters read the same poll and revise at the same time, so the next profile is $F(S)_i=\operatorname{BR}_i(n-\mathbf e_{s_i}).$
For a fixed $p>0$, let $F_p$ denote the analogous update with exact expected payoffs. The next lemma shows that, under strict utilities, both updates are well defined and coincide for small $p$, without any further assumption on utilities.

\begin{lemma}[Unique limiting responses] \label{lem:small-p}
    Under strict utilities and uniform tremble destinations, every private vector has a unique limiting best response. For each fixed finite election, $F_p=F$ on all profiles for all sufficiently small positive $p$.
\end{lemma}
\begin{proof}
At $p_\star=(m-1)/m$, every other ballot is uniform on $M$, regardless of its current value. By symmetry, casting $c$ gives $c$ an expected lottery share $\alpha$ (the expectation of $\mathbf 1\{c\in W\}/|W|$) and every other candidate the same share $\beta$. Moreover, $\alpha>\beta$: moving the focal ballot from $c'$ to $c$ never decreases $c$'s share, and increases it strictly at terminal counts $(k,k,0,\ldots,0)$ or $(k,k+1,0,\ldots,0)$ in $(c,c',\ldots)$ coordinates, according as $|q|=2k$ or $2k+1$. Such counts have positive probability. Hence
\[
U_i^{p_\star}(c\mid q)-U_i^{p_\star}(c'\mid q)
=(\alpha-\beta)\bigl(u_i(c)-u_i(c')\bigr)\ne0.
\]
No pairwise polynomial is identically zero. Its sign is constant near zero, so the finitely many ballots have a strict ordering there. Taking the minimum of the finitely many bounds over voters and profiles proves the second claim.
\end{proof}

Thus $F$ is a deterministic map on the finite set of complete profiles, and every trajectory is eventually periodic. A lowest-order coefficient can vanish for special utilities; then the first later nonzero coefficient decides the comparison, so no generic-utility assumption is needed. The bound on $p$ can depend on the entire election and on $N$; we do not claim a uniform bound as electorates grow.

\paragraph{Convergence and cycles.}
Initialization is \emph{truthful} if every voter supports her favorite candidate. A profile is \emph{Duvergerian} if at most two candidates have positive support, and \emph{stable} if $F(S)=S.$
A stable profile is thus precisely a pure trembling-hand equilibrium in the sense of Obraztsova et al.\ \cite{obraztsova2016trembling}.
The dynamics \emph{converge} if they reach a stable profile. 
Otherwise, the trajectory enters a \emph{recurrent orbit} after finitely many rounds, its \emph{preperiod}. The \emph{period} is the smallest positive return time of the complete profile on this orbit, and an \emph{active} voter changes ballot somewhere on it. The period of the polls is the smallest positive $L$ for which $n^{t+L}=n^t$ at every recurrent phase; it divides the period of the profile.

With $m=2$, truthful voting is trivially stable, so the dynamics converge immediately; we assume $m\geq3$ from now on. We use $c,c'\in M$ for arbitrary candidates and $A,B,C$ for named candidates. For a single voter's ranking with three candidates, $X\succ Y\succ Z$ denotes her favorite, middle, and least-preferred candidates.

\section{One-Step Convergence to Duverger}\label{sec:fixation}

That such dynamics often converge, and towards Duverger, is perhaps not surprising. A voter compares ballots on the most likely configurations in which her ballot matters. If her favorite trails far behind, many trembles are needed before a ballot for that candidate could make a difference, whereas a few trembles already make the race between the two leaders close. Her decisive comparison should therefore involve only the leaders, and she should vote for the one she prefers. 
Two things could still go wrong though. First, the revised profile could trigger further revisions, as in the introductory story, where a narrowly trailing candidate gains and loses supporters at once; a large enough gap rules this out. Second, the leaders themselves could lose support; truthful initialization prevents this, since their supporters already vote for them and keep doing so. We formalize this argument in the remainder of the section.

\subsection{Arbitrary Number of Candidates} \label{sec:general-m}

We first consider an election with any number of candidates and a general tremble belief, and derive a simple sufficient condition for one-step convergence. 
In this subsection only, voters may believe that trembles are biased towards some candidates. A ballot for $c$ trembles to $c'\ne c$ with probability $p\lambda_{cc'}$, where the shares $\lambda_{cc'}>0$ sum to one over $c'\ne c$ and do not depend on $p$. Trembles remain independent, and uniform trembles are the case $\lambda_{cc'}=1/(m-1)$; expected payoffs are defined as before.
The following lemma then shows that a voter facing two leaders well separated from all other candidates has a unique best response---her preferred leader---independent of her tremble belief.

\begin{lemma}[Choice between separated leaders]
\label{lem:general-two-candidates}
If $A,B$ have the two highest private counts and
$\min\{q_A,q_B\}\geq q_c+3$, for every $c\notin\{A,B\},$
then the unique limiting response is the one of $A,B$ the voter prefers, for every choice of positive shares $\lambda$.
\end{lemma}
\begin{proof}
Relabel so that $q_A\geq q_B$ and put $d=q_A-q_B$, $r=\lceil d/2\rceil$. 
In any configuration where at most $r$ other ballots tremble, the resulting private vector $v$ satisfies
\[
v_A-v_c\geq d+3-2r\geq2
\qquad(c\notin\{A,B\}).
\]
Even after adding any contemplated ballot, no lower candidate can win or tie for the lead. On every such terminal assignment, voting for the preferred member of $\{A,B\}$ is therefore weakly better than any other ballot: it can only shift the winner lottery toward that candidate.

There is a strict comparison at or before order $r$. If the voter prefers $A$, use $r$ trembles $A\mapsto B$, leaving private gap $0$ when $d$ is even and $-1$ when $d$ is odd. Ballot $A$ then strictly beats every alternative. If she prefers $B$, use $\lfloor d/2\rfloor$ trembles $A\mapsto B$, leaving gap $0$ or $1$, and ballot $B$ strictly beats every alternative. These flows are feasible because $q_A\geq d+3$. Every assignment with $\eta$ trembles has probability $p^\eta(1-p)^{N-1-\eta}$ times a positive constant depending on the assignment. Take the earliest order with a strict contribution to the required margin. It has $\eta\leq r$, all contributions at that order are nonnegative, and all earlier ones vanish. Their positive sum is therefore the first nonzero coefficient.
\end{proof}

The lemma concerns a single voter's private counts, while all voters respond to the same public poll. Since removing one's own ballot lowers a gap by at most one, a public gap of four gives every voter a private gap of three. Combined with truthful initialization, this yields convergence in a single update.

\begin{theorem}[One-step convergence] \label{thm:general-gap-fixation}
    Let $S^0$ be truthful and let $A,B$ be the two leaders of the initial poll.
    If $   \min\{n_A^0,n_B^0\}\geq n_c^0+4$, for every $c\notin\{A,B\}$,
    then one update makes every voter support her preferred member of $\{A,B\}$, and the resulting profile is stable. This holds for every choice of positive shares $\lambda$.
\end{theorem}
\begin{proof}
Removing any ballot leaves both leaders at least three above every lower candidate, so Lemma~\ref{lem:general-two-candidates} determines every first response. Truthful $A$- and $B$-supporters keep their ballots. Both candidates therefore retain at least their initial support, which is at least four, while every other candidate has zero support. At the successor, removing any ballot again leaves both supported candidates at least three above all others, and the lemma makes every voter stay.
\end{proof}

After a single revision, the profile thus coincides with sincere voting in a runoff between the two poll leaders, although no candidate is formally eliminated. Note that the shares $\lambda$ must not vanish with $p$: otherwise, trembles towards some candidates become less likely than others, which changes which configurations are decisive.


\subsection{Sharper Conditions for Three Candidates}\label{sec:three-candidates}

The previous result gives a sufficient condition in a very general setting. We next show that this condition can be relaxed considerably assuming three candidates and uniform trembles. We first state a simple but useful lemma.

\begin{lemma}[Least-preferred ballot]\label{lem:top-bottom-dominance}
    No voter's limiting best response is her least-preferred candidate. With three candidates, each voter thus chooses between her favorite and middle candidates.
\end{lemma}

\begin{proof}

Let $X$ and $Z$ be the voter's favorite and least-preferred candidates. Fix a terminal private vector $v$, and compare $W_X=W(v+\mathbf e_X)$ with $W_Z=W(v+\mathbf e_Z)$. 
Since the two ballots change only the counts of $X$ and $Z$, either $W_X=\{X\}$, or $W_Z=\{Z\}$, or $W_X\setminus\{X\}=W_Z\setminus\{Z\}$. In the first two cases, $u_i(W_X)\geq u_i(W_Z)$ is immediate. In the third, $W_X$ is obtained from $W_Z$ by possibly removing $Z$ and possibly adding $X$; neither operation lowers the value of the uniform lottery. Hence ballot $X$ is weakly better at every terminal vector. It is strictly better at $v=(k,k,0,\ldots,0)$ or $v=(k,k+1,0,\ldots,0)$ in $(X,Z,\ldots)$ coordinates, according as $|q|=2k$ or $2k+1$. These vectors have positive probability, so the expected comparison is strict for every $p>0$.

\end{proof}
This is the trembling-hand counterpart of the classical fact that voting for one's worst alternative is weakly dominated under plurality~\cite{dhillon_when_2004}.
With three candidates, the gap that matters is the one between the runner-up and the trailing candidate. The next theorem shows that a gap of two always suffices, and a gap of one suffices when the leaders' gap is even. Parity matters because a tremble from one leader to the other changes their gap by two, so whether they can be brought exactly level depends on its parity.

\begin{theorem}[Three-candidate refinement]\label{thm:condensed-one-step}
    Let $S^0$ be truthful, with $n_A^0\geq n_B^0>n_C^0\geq1$. If $n_B^0-n_C^0\geq2$ or
    $n_A^0-n_B^0$ is even,
    then initial $A$- and $B$-voters stay, initial $C$-voters choose their preferred member of $\{A,B\}$, and $S^1$ is stable.
\end{theorem}
\begin{proof}
Write the poll as $(r+d,r,r-s)$, where $r\geq2$, $d\geq0$, and $1\leq s\leq r-1$. The three truthful voter blocks see
\begin{gather*}
    q^{A}=(r+d-1,r,r-s), \quad q^{B}=(r+d,r-1,r-s), \\
    \text{and} \quad q^{C}=(r+d,r,r-s-1).
\end{gather*}

\noindent By Lemma~\ref{lem:top-bottom-dominance}, each voter has two candidate ballots. A tremble $A\mapsto B$ reduces the $A$--$B$ gap by two, and $A\mapsto C$ or $C\mapsto B$ by one, so parity determines which flows can first make a ballot decisive. At $d=0$ the claimed responses follow from the tied-top winner sets; for $d>0$, Table~\ref{tab:fixation-orders} gives the order of each required positive margin, and all earlier coefficients vanish.

\begin{table}[t]
\caption{Orders of the first positive margins in the proof of Theorem~\ref{thm:condensed-one-step}, as powers of $p$. Only favorite--middle comparisons are needed, by Lemma~\ref{lem:top-bottom-dominance}.}
\label{tab:fixation-orders}
\centering
\begin{tabular}{cccc}
\toprule
Ranking & Margin & $d=2k>0$ & $d=2k+1,\ s\geq2$\\
\midrule
$A\succ B\succ C$ & $A-B$ & $k-1$ & $k$\\
$A\succ C\succ B$ & $A-C$ & $k$ & $k$\\
$B\succ A\succ C$ & $B-A$ & $k$ & $k+1$\\
$B\succ C\succ A$ & $B-C$ & $k$ & $k+1$\\
$C\succ A\succ B$ & $A-C$ & $k$ & $k+1$\\
$C\succ B\succ A$ & $B-C$ & $k$ & $k$\\
\bottomrule
\end{tabular}
\end{table}

The even-gap boundary $s=1$, with $d=2k>0$, shows why the guarantee depends on parity. A truthful $B\succ C\succ A$ voter sees $q^B=(r+2k,r-1,r-1)$. At order $k$, the only contributing flows are $k(A\mapsto B)$ and $k(A\mapsto C)$, with equal multiplicity and payoff differences $(u_i(B)-u_i(A))/2$ and $(u_i(A)-u_i(C))/2$. Their sum gives 
\[
2^{-k-1}\binom{r+2k}{k}\bigl(u_i(B)-u_i(C)\bigr)>0
\]
as the first nonzero coefficient. 
Thus even when the trailing candidate enters the first decisive order, its contribution preserves the required response. For an odd gap, $s=2$ is the corresponding boundary: the trailing candidate can again enter the first decisive order, so its contribution must be summed. For a truthful $B\succ C\succ A$ voter, put $a=r+2k+1$. The coefficient of $p^{k+1}$ in $U_i^p(B\mid q^B)-U_i^p(C\mid q^B)$ is
\[
\begin{split}
2^{-k-2}\Bigl[\binom{a}{k+1}\bigl(u_i(B)-u_i(C)\bigr)
+(2r+k-1)\binom{a}{k}\bigl(u_i(B)-u_i(A)\bigr)\Bigr],
\end{split}
\]
which is positive since $u_i(B)$ exceeds both $u_i(C)$ and $u_i(A)$.
The remaining rows follow by the same gap reduction and flow counts; Appendix~\ref{app:one-step-blocks} gives the flow lists and the remaining coefficient and sign calculations.

It remains to show that $S^1$ is stable. If $x$ of the initial $C$-voters move to $A$, the successor counts are $(r+d+x,2r-s-x,0)$, and every voter prefers the candidate she supports to the other supported candidate. For $r\geq3$, both supported counts are at least three. If both private counts are at least three, Lemma~\ref{lem:general-two-candidates} makes the voter stay. The only remaining case, up to exchanging $A,B$, is an $A$-supporter with $q_A=2$ and $q_B\geq3$. With at most $\lfloor(q_B-2)/2\rfloor$ trembles, $B$ remains at least two above $C$, so the unsupported candidate cannot win after any contemplated ballot. Ballot $A$ is weakly best in every such configuration, and moving exactly that many ballots $B\mapsto A$ makes it strictly best. The earliest-order argument of Lemma~\ref{lem:general-two-candidates} therefore makes this voter stay as well.

If $r=2$, then $s=1$, and the hypothesis forces $d=2k$ to be even. The successors are $(2k+2,3,0)$ and $(2k+3,2,0)$. The first is covered above for $k\geq1$. At $(2,3,0)$, both $A$-supporters are truthful $A$-voters; their first $A$-minus-$C$ coefficient is $3[2u_i(A)-u_i(B)-u_i(C)]/4>0$, and their $A$-minus-$B$ coefficient is positive. At $(2k+3,2,0)$, every $B$-supporter is a truthful $B$-voter. Her first $B$-minus-$C$ coefficient is a positive multiple of $(k+1)(u_i(B)-u_i(A))+u_i(B)-u_i(C)>0$, while her $B$-minus-$A$ comparison has the sign of her preference. An $A$-supporter has an even private $A$--$B$ gap $2k$; its first decisive order contains only supported winners and favors $A$. Hence both boundary successors are stable. The full coefficients are given in Appendix~\ref{app:one-step-fixedness}.
\end{proof}

The theorem leaves open tied lower counts and a one-vote gap below the runner-up when the leaders' gap is odd.
The latter can genuinely fail to converge in one step, much as in the story of the introduction. Take four voters with utilities $(1,1/2,0)$, three with $(19/20,1,0)$, and two with $(0,19/20,1)$, in $(A,B,C)$ order. Starting from truthful votes, the dynamics give
\[
(4,3,2)\mapsto(7,2,0)\mapsto(4,5,0).
\]
In the first update, the $C$-voters rally behind $B$, while the $B$-voters, who prefer $B$ only slightly to $A$ and see $C$ close behind, defect to $A$: at their private counts $(4,2,2)$, their margin for $B$ over $A$ is $-19p/40+\mathcal O(p^2)$. $B$ is now carried entirely by former $C$-voters. With $C$ gone, the defectors face $(6,2,0)$, where their margin becomes $3p^2/8+\mathcal O(p^3)>0$, and they return. The final profile is stable, and $B$ wins although her own supporters had abandoned her. Here the dynamics still converge after two updates; Section~\ref{sec:cycles} shows that they need not converge at all.

\subsection{Duverger's Law in Typical Elections}\label{sec:random-elections}

We now argue that these results offer an intuitive explanation for why Duverger's law is commonly observed in practice. Fix $m\geq3$ candidates and a probability distribution over utility vectors, with almost surely strict rankings, and view each observed electorate of $N$ voters as $N$ independent draws from it. Let $\alpha_c$ be the probability that candidate $c$ is a voter's favorite, and assume $\alpha_c>0$ for every $c$. The truthful poll is then multinomial with parameters $N$ and $(\alpha_c)_{c\in M}$. The condition of Theorem~\ref{thm:general-gap-fixation} fails only if the runner-up is within three votes of the third-placed candidate, and such near-ties become rare as the electorate grows.

\begin{corollary}[Random elections]\label{cor:general-m-fixation}
    In this random model, truthful initialization reaches a stable Duvergerian profile after one update with probability at least $1-\mathcal O(N^{-1/2})$. If two candidates have favorite probabilities strictly greater than those of every other candidate, the probability is at least $1-e^{-\Omega(N)}$.
\end{corollary}
\begin{proof}
If $\alpha_c=\alpha_{c'}$, condition on $K=n_c^0+n_{c'}^0$. Then $n_c^0$ is binomial with parameters $K,1/2$, and $|n_c^0-n_{c'}^0|\leq3$ permits at most four values. Outside the exponentially unlikely event $K<N\alpha_c$, the largest conditional point probability is $\mathcal O((N\alpha_c)^{-1/2})$. Thus the near-tie probability is $\mathcal O(N^{-1/2})$, including the possibility $K=0$.
If $\alpha_c\ne\alpha_{c'}$, concentration about the nonzero mean gap gives an exponentially small probability. A union bound over the fixed number of pairs shows that all counts differ by at least four with probability at least $1-\mathcal O(N^{-1/2})$, and Theorem~\ref{thm:general-gap-fixation} applies on this event. For the stronger bound, let $A,B$ be the two candidates whose favorite probabilities exceed all others. Each difference $n_A^0-n_c^0$ or $n_B^0-n_c^0$, with $c\notin\{A,B\}$, has a positive mean proportional to $N$. Hoeffding's inequality bounds the probability that any one is below four by $e^{-\Omega(N)}$, and a union bound over these pairs proves the claim; no separation between $A$ and $B$ is needed.
\end{proof}

The bound holds for every such distribution; only the constant depends on $m$ and on the distribution, and it grows as the favorite probabilities of the trailing candidates shrink.
Near-ties are exponentially unlikely unless the second-highest favorite probability is shared by several candidates, as when the distribution is symmetric across candidates. Even then, they are rare, so elections rarely fail to settle on two candidates after a single round of revisions.
Near-ties matter mainly when candidates are equally popular, for instance under a distribution symmetric across candidates, with $\alpha_c=1/m$ for all $c$. Even then, they are rare, so elections rarely fail to settle on two candidates after a single round of revisions.

\section{When the Dynamics Do Not Converge} \label{sec:cycles}

The results above show that convergence is the typical outcome, but it is not guaranteed. The following example, inspired by an election of Obraztsova et al.~\cite{obraztsova2016trembling} that admits no trembling-hand equilibrium, provides a simple counter-example.

\begin{example}[A period-two cycle]\label{ex:two-cycle}
    Consider three candidates and four voters: voters $1$ and $2$ rank \cA${}\succ{}$\cB${}\succ{}$\cC, voter $3$ ranks \cB${}\succ{}$\cC${}\succ{}$\cA, and voter $4$ ranks \cC${}\succ{}$\cB${}\succ{}$\cA, with any strict utilities consistent with these rankings. Starting from truthful votes, the ballot profile alternates between
    \[
    (\cA,\cA,\cB,\cC)\;\longleftrightarrow\;(\cA,\cA,\cC,\cB),
    \]
    where the $i$-th entry is voter $i$'s ballot. Voters $3$ and $4$ each try to create a tie between $A$ and their favorite by switching to the other's candidate, and their simultaneous switches cancel out.
\end{example}

This cycle is hardly a failure of coordination. The poll reads $(2,1,1)$ in both phases; only two voters exchange their ballots. Judged by the poll alone, these dynamics have converged. Can the dynamics exhibit genuinely richer behavior, with longer cycles in which the poll itself keeps changing?

\subsection{Three-Cycles: How Many Active Voters?}\label{sec:3-cycles}

We now look for richer cycles, in which the poll itself changes from round to round, and start with simple ones: three candidates and period three. We first show that such cycles require enough \emph{active} voters, that is, voters who change their ballot somewhere on the cycle; all other voters keep the same ballot throughout, and their counts act as a fixed background.

\begin{theorem}[Active-voter lower bound]
\label{thm:no-three-active-period-three}
With three candidates, every period-three cycle of the complete ballot profile has at least four active voters.
\end{theorem}

\begin{proof}[Proof outline]
If the other voters cast the same ballots at two phases, the focal voter has the same private counts and must repeat her response. By Lemma~\ref{lem:top-bottom-dominance}, each active voter alternates between her favorite and her middle candidate. Up to a time shift, her ballots along the cycle are therefore $(X,X,Y)$ or $(Y,Y,X)$, with exactly one stay among the three transitions.

A single active voter always faces the same private vector and cannot cycle. With two, suppose voter $1$ stays in the first transition. Voter $2$ then faces the same other ballots at the first two phases, so she stays in the next transition. This forces voter $1$ to stay again, a contradiction. With three active voters, the same argument, and the fact that not all voters can stay in the same transition, leave only two patterns of stays, $(1,1,1)$ and $(2,1,0)$, up to a time shift.

For these patterns, let $z=(z_A,z_B,z_C)$ count the ballots of the inactive voters. Two restrictions on responses rule them out. First, a voter with $X\succ Y\succ Z$ can choose $Y$ only if $q_Y>q_X$. 
Second, for $c\in\{Y,Z\}$ with $q_c>0$, no voter responds $X$ at $q$ and $Y$ at $q+\mathbf e_X-\mathbf e_c$: moving one other ballot towards a voter's favorite cannot make her abandon it.
For the pattern $(1,1,1)$, these restrictions leave two classes of cycles, each requiring an impossible strict cyclic chain among $z_A,z_B,z_C$. For $(2,1,0)$, six classes remain, each violating one of the restrictions. Appendix~\ref{supp:active-proof} proves both restrictions and covers every branch of the case analysis.
\end{proof}

The bound holds for any number and composition of inactive voters, and it does not require truthful initialization. As we show next, six active voters suffice: we exhibit a period-three cycle, reached from truthful votes, in which the poll changes in every round.

\begin{figure*}[t]
\centering
\colorlet{paperA}{paperBlue}\colorlet{paperB}{paperGreen}\colorlet{paperC}{paperRed}%
\newcommand{\fourballot}[1]{\colorbox{paper#1}{\makebox[1.1em]{\strut$\mathsf{#1}$}}}%
\newcommand{\fourchanged}[1]{\fcolorbox{black}{paper#1}{\makebox[1.1em]{\strut$\mathsf{#1}$}}}%
\setlength{\tabcolsep}{2pt}%
\setlength{\fboxsep}{2pt}%
\renewcommand{\arraystretch}{1.3}%
\begin{minipage}[t]{0.49\textwidth}
  \centering
  \parbox[c][4.6cm][c]{\linewidth}{\centering
  \resizebox{\linewidth}{!}{%
  \footnotesize\begin{tabular}{cc*{5}{c}}
  \toprule
  $i$ & $(u_i(A),u_i(B),u_i(C))$ & $0$ & $1$ & $2$ & $3$ & $4$\\
  \midrule
  $1$ & $(1,\frac23,0)$ & \fourballot{A} & \fourballot{B} & \fourballot{A} & \fourballot{A} & \fourballot{A}\\
  $2$ & $(1,\frac23,0)$ & \fourballot{A} & \fourballot{B} & \fourballot{A} & \fourballot{A} & \fourballot{A}\\
  $3$ & $(\frac14,1,0)$ & \fourballot{B} & \fourballot{B} & \fourballot{B} & \fourballot{B} & \fourballot{B}\\
  $4$ & $(\frac56,1,0)$ & \fourballot{B} & \fourballot{A} & \fourchanged{B} & \fourchanged{A} & \fourballot{B}\\
  $5$ & $(0,\frac23,1)$ & \fourballot{C} & \fourballot{B} & \fourchanged{C} & \fourchanged{B} & \fourballot{C}\\
  $6$ & $(\frac14,0,1)$ & \fourballot{C} & \fourballot{C} & \fourchanged{A} & \fourchanged{C} & \fourballot{C}\\
  \midrule
  \multicolumn{2}{c}{Poll} & $(2,2,2)$ & $(1,4,1)$ & $(3,2,1)$ & $(3,2,1)$ & $(2,2,2)$\\
  \multicolumn{2}{c}{Winners} & $\{A,B,C\}$ & $\{B\}$ & $\{A\}$ & $\{A\}$ & $\{A,B,C\}$\\
  \bottomrule
  \end{tabular}}}
  \caption{A period-four cycle with alternating leaders. Rows are voters, columns are rounds, and round $4$ repeats round $0$. Outlined ballots change between rounds $2$ and $3$ although the poll does not: responses depend on each voter's own ballot.}
  \Description{Six voters' utilities and ballots over four updates. The profiles are AABBCC, BBBABC, AABBCA, AABABC, and AABBCC. The third and fourth profiles share the poll (3,2,1) but differ in voters 4, 5, and 6. B and A each recur as the unique plurality leader.}
  \label{fig:truthful-four-cycle}
\end{minipage}\hfill
\begin{minipage}[t]{0.49\textwidth}
  \centering
  \parbox[c][4.6cm][c]{\linewidth}{\centering
  \resizebox{\linewidth}{!}{%
  \begin{tikzpicture}[x=0.52cm,y=0.42cm,font=\footnotesize]
    \fill[paperBlue!70] (3,0) rectangle (8,9.6);
    \fill[paperGreen!70] (8,0) rectangle (13,9.6);
    \node[anchor=north] at (5.5,9.55) {First half};
    \node[anchor=north] at (10.5,9.55) {Second half};
    \draw[->] (0,0) -- (13.6,0) node[right] {round};
    \draw[->] (0,0) -- (0,10.1) node[above] {$x$};
    \foreach \t in {0,3,8,13} \draw (\t,0) -- (\t,-0.15) node[below] {$\t$};
    \foreach \y in {0,2,4,6,8} \draw (0,\y) -- (-0.15,\y) node[left] {$\y$};
    \draw[thick] (0,9) -- (1,0) -- (2,8) -- (3,2) -- (4,5) -- (5,4) -- (6,3) -- (7,6)
                  -- (8,2) -- (9,5) -- (10,4) -- (11,3) -- (12,6) -- (13,2);
    \foreach \t/\x in {0/9,1/0,2/8} \fill[white,draw=black] (\t,\x) circle (2.2pt);
    \foreach \t/\x in {3/2,4/5,5/4,6/3,7/6,8/2,9/5,10/4,11/3,12/6,13/2} \fill (\t,\x) circle (2.2pt);
  \end{tikzpicture}}}
  \caption{The long cycle for $\ell=5$ ($27$ voters), with poll $(x,x,27-2x)$ for $x$ split pairs. Open dots mark the transient updates. The poll has period five, the profile ten: the shaded halves differ in which pairs are split.}
  \Description{A line plot of the number of split pairs x against the round. Starting from 9, the trajectory goes to 0 and 8, then cycles through 2, 5, 4, 3, 6 twice. The two repetitions are shaded differently because the identities of the split pairs differ.}
  \label{fig:long-cycle}
\end{minipage}
\end{figure*}

\begin{theorem}[Six active voters suffice] \label{ex:six-active-cycle}
    There is a three-candidate election with six voters, all of them active, whose dynamics from truthful votes cycle with period three.
\end{theorem}
\begin{proof}
Choose $9/11<\rho<1$ and $1/2<\gamma<1$. Take two voters with ranking $A\succ C\succ B$, two with $B\succ C\succ A$, one with $C\succ A\succ B$, and one with $C\succ B\succ A$. Their utilities in ranking order are $(1,\rho,0)$ for the first four and $(1,\gamma,0)$ for the last two. Their ballots over three rounds are $(A,C,C)$, $(B,C,C)$, $(C,A,C)$ and $(C,B,C)$, respectively, so the counts are
$$ (2,2,2)\mapsto(1,1,4)\mapsto(0,0,6)\mapsto(2,2,2).$$

Appendix~\ref{app:six-active-details} checks, for each round and each voter, that the next ballot on her path is her strict best response: the first nonzero coefficient of the payoff difference with her other option is positive throughout these intervals.
 Every voter returns to her ballot after three updates, and the three count vectors are distinct, thus the period is exactly three.
\end{proof}

The cycle has a simple reading. From the three-way tie, the supporters of $A$ and $B$ move to their common middle candidate $C$, while the $C$-voters simultaneously move to their own middle candidates; $C$ thus leads, but without its own supporters. These return in the next round, so that everyone votes for $C$. The supporters of $A$ and $B$ then return to their favorites, and the tie reappears. 
The minimum number of active voters for period three thus lies between four and six. An extensive computer search found no period-three cycle with four or five active voters and at most ten inactive voters per candidate.

\subsection{Cycles with Changing Leaders}\label{sec:changing_leaders}

In the period-three cycle above, the poll changes in every round, but the outcome hardly does: $C$ leads alone in two rounds, and the third is a three-way tie. A more convincing failure of coordination is a cycle in which different candidates take turns as the unique leader. The following example showcases this property.

\begin{example}[Alternating leaders]\label{ex:truthful-four-cycle}
Figure~\ref{fig:truthful-four-cycle} lists the utilities and ballots of six voters. Starting from truthful votes, the dynamics cycle with period four, and both $B$ and $A$ recur as unique plurality leaders. Voter $3$ always votes for $B$; the other five voters are active.
\end{example}

From the three-way tie, the voters whose second choice is $B$ rally behind it, while the $B$-voter who nearly likes $A$ as much moves to $A$. Facing a strong $B$, the $A$-voters return and are joined by the $C\succ A\succ B$ voter, so $A$ takes the lead. The poll then stays at $(3,2,1)$ for two rounds while three voters exchange ballots, before the tie reappears. 
We prove in Appendix~\ref{app:four-cycle-margins} that every response is a strict best response; as the four profiles are distinct, the period is 
four.

Moreover, the two rounds with the same poll show that a voter's response depends on her own current ballot. Consider voter $4$, with ranking $B\succ A\succ C$. At round $2$, she votes for $B$ and sees private counts $(3,1,1)$: both ballots leave $A$ ahead, so only trembles matter. The one-tremble flows $A\mapsto B$, $A\mapsto C$ and $C\mapsto B$, with multiplicities $3$, $3$ and $1$, change her payoff from voting $A$ rather than $B$ by $-1/6$, $2/9$ and $-1/12$, respectively, giving
\[
U_4^p(A\mid(3,1,1))-U_4^p(B\mid(3,1,1))=\frac{p}{24}+\mathcal O(p^2)>0,
\]
so she switches to $A$. At round $3$, the poll is unchanged, but since she now votes for $A$, her private counts are $(2,2,1)$. Voting $B$ now makes her favorite win outright, so her margin for $B$ over $A$ is $1/6+\mathcal O(p)>0$, and she switches back.

\subsection{Arbitrarily Long Cycles}
\label{sec:long-cycles}
All cycles so far are short and involve at most six active voters. We now show that the dynamics can be much richer: with only three candidates, larger electorates sustain cycles of arbitrarily long period, in which the poll keeps changing round after round.

\begin{theorem}[Cycle lengths growing with the electorate]
\label{thm:long-period-family}
    For every odd integer $\ell\geq 3$, there is a nonempty open set of strict utility profiles with three candidates and $N=3(\ell+4)$ voters, for which the dynamics from truthful votes enter a recurrent orbit after three updates. On this orbit, the poll has period $\ell$, whereas the complete ballot profile has period $2\ell$. Every candidate keeps positive support throughout, and $C$ is the unique plurality leader.
\end{theorem}

Figure~\ref{fig:long-cycle} shows the construction for $\ell=5$, with $27$ voters. After three transient updates, the poll cycles through five distinct values. When it returns, however, two pairs of voters have exchanged their roles, so the complete profile returns only after ten updates.

The construction rests on pairs of opposing partisans who share a compromise candidate. Each pair consists of one voter with ranking $A\succ C\succ B$ and one with $B\succ C\succ A$; a pair is \emph{split} when its voters support their favorites $A$ and $B$, and \emph{pooled} when both vote for $C$. Loyal $C$-supporters complete the electorate. Since every pair is either split or pooled, the poll is determined by the number $x$ of split pairs, and $A$ and $B$ are always tied. Whether a voter keeps her favorite or retreats to $C$ then depends on the parity of $x$ and on how strongly she values the compromise. Odd values of $x$ flip every pair at once, while even values let only the pairs that value the compromise least move towards their favorites. Choosing these utilities carefully lets us steer the number of split pairs through a long sequence of distinct values before it repeats.

Formally, fix an odd integer $h=2s+1\geq7$. For each identity $i$ in a set $P$ of $h$ pairs, create one voter with ranking $A\succ C\succ B$ and one with ranking $B\succ C\succ A$. After normalizing top and bottom utilities to $1$ and $0$, write their utilities for $C$ as $\rho_i^A$ and $\rho_i^B$. Add $h$ voters with ranking $C\succ A\succ B$ and utility $\gamma_j$ for $A$. A profile in which every pair is split or pooled is identified by its set $I\subseteq P$ of split pairs. Writing $x=|I|$, its poll is
\begin{equation}
(x,x,3h-2x).
\label{eq:paired-counts}
\end{equation}
Let $\Phi(I)$ denote the set of split pairs after one update. A pair's response depends on how much its voters value the compromise $C$: the more they value it, the more readily they pool. The lemma makes this precise through a single increasing chain of utility thresholds, and classifies pairs by where their utility for $C$ falls in this chain.

\begin{lemma}[Paired response]\label{lem:paired-response}
    There exist thresholds $\tfrac12=\kappa_0<\kappa_1<\cdots<\kappa_h<\kappa_{h+1}=1$ with the following property. Whenever $x=2k$ pairs are split, a voter in a split pair keeps her favorite if her utility for $C$ is strictly below $\kappa_k$, and moves to $C$ if it is strictly above; a voter in a pooled pair returns to her favorite if her utility for $C$ is strictly below $\kappa_{h-k}$, and stays at $C$ if it is strictly above.

    Consequently, suppose that each pair $i$ has a \emph{type} $\tau_i\in\{0,\ldots,h\}$ with $\rho_i^A,\rho_i^B\in(\kappa_{\tau_i},\kappa_{\tau_i+1})$, and that every $\gamma_j\in(0,1/2)$. Then every pair remains split or pooled, all responses are strict, the added voters remain at $C$, and
    \begin{equation}\label{eq:paired-map}
        \Phi(I)=
        \begin{cases}
        P\setminus I,&x\text{ odd},\\[2pt]
        \{i\in I:\tau_i<k\}\cup
        \{i\in P\setminus I:\tau_i<h-k\},&x=2k.
        \end{cases}
    \end{equation}
\end{lemma}

\begin{proof}

Consider an $A\succ C\succ B$ voter with normalized utilities $(u(A),u(C),u(B))=(1,\rho,0)$, where $1/2<\rho<1$. Let $d :=h-x$ be the number of pooled pairs, so that $C$ leads both $A$ and $B$ by $3d$ votes. Her private counts are
\[
\begin{aligned}
q^A&=(x-1,x,x+3d)&&\text{if her pair is split, so }x>0,\\
q^C&=(x,x,x+3d-1)&&\text{if her pair is pooled, so }x<h.
\end{aligned}
\]
By Lemma~\ref{lem:top-bottom-dominance}, she never votes $B$.

\emph{Which trembles matter.} At a terminal private vector $v$, the comparison between $A$ and $C$ is zero whenever $v_A\leq v_C-2$ and $v_B\leq v_C-1$: both ballots then leave $C$ the unique winner. A nonzero contribution therefore requires
\[
v_A\geq v_C-1\qquad\text{or}\qquad v_B\geq v_C.
\]
Each tremble reduces either lead of $C$ by at most two. To reduce a lead by $2r$ with $r$ trembles, every tremble must move a ballot from $C$ to the corresponding candidate. To reduce it by $2r-1$, either all $r$ trembles have this direction, or exactly one reduces the lead by one. For the $C$--$A$ lead, the latter directions are $C\mapsto B$ and $B\mapsto A$; for the $C$--$B$ lead they are $C\mapsto A$ and $A\mapsto B$. This gives the complete lists of first contributing flows below, and shows that all lower orders contribute nothing.

\emph{Odd number of split pairs.} If $x$ is odd, $d$ is even. For $d=0$, a split voter has margin $(1-2\rho)/3<0$ without trembles, so she chooses $C$. For $d\geq2$, put $t=3d/2$. At $q^C$, the first contributing order is $t-1$, with the single flow $(t-1)(C\mapsto A)$ and payoff difference $(1-\rho)/2>0$. At $q^A$, the first order is $t$, with flows $t(C\mapsto A)$ and $t(C\mapsto B)$ of equal multiplicity and payoff differences $(1-\rho)/2$ and $-\rho/2$, whose sum is negative. Thus every pair changes configuration, proving the odd case of \eqref{eq:paired-map}.

\emph{Even number of split pairs: pooled voters.} Now let $x=2k$, so $d$ is odd. First take $d\geq3$ and put $t=(3d-1)/2$. A pooled voter sees $q^C=(x,x,x+2t)$, and her first contributing order is $t$. After dividing multiplicities by $\binom{x+2t}{t}$, the complete list is
\[
\begin{array}{c|c|c}
\text{Flow}&\text{Weight}&\text{Payoff: }A-C\\ \hline
t(C\mapsto A)&1&1-\rho\\
(t-1)(C\mapsto A)+(C\mapsto B)&t&(1-\rho)/2\\
t(C\mapsto B)&1&-\rho/2\\
(t-1)(C\mapsto A)+(B\mapsto A)&\dfrac{tx}{x+t+1}&(1-\rho)/2.
\end{array}
\]
For example, the second flow has multiplicity $t\binom{x+2t}{t}$, while the fourth has multiplicity $x\binom{x+2t}{t-1}$. Each order-$t$ assignment carries the same probability factor $2^{-t}$, so summing the table gives, up to a positive factor,
\[
2+t+\frac{tx}{x+t+1}-\left(3+t+\frac{tx}{x+t+1}\right)\rho.
\]
Its root defines $\kappa_{h-k}$; away from equality, a pooled voter chooses $A$ exactly below this root.

\emph{Even number of split pairs: split voters.} If $x>0$, a split voter sees $q^A=(x-1,x,x+2t+1)$, and her first contributing order is $t+1$. The same gap reductions give two pure and four mixed flows, whose sum has the sign of
\[
t+3+\frac{x(t+1)}{x+t+1}
-\left(2t+5+\frac{(2x-1)(t+1)}{x+t+1}\right)\rho.
\]
Its root defines $\kappa_k$; away from equality, a split voter keeps $A$ exactly below it. Appendix~\ref{app:paired-response} lists all six flows and their multiplicities. When $x=0$, no split voter exists, and $\kappa_0=1/2$ is only the lower utility bound.

\emph{The case of a single pooled pair.} For $d=1$, three-way ties change the calculation. The pooled voter's order-one flows are $C\mapsto A$, $C\mapsto B$ and $B\mapsto A$. For the split voter, the order-two mixed flow $(C\mapsto A)+(C\mapsto B)$ produces a three-way tie, with payoff difference $(1-2\rho)/3$, and must be counted once. Summing the flow lists in Appendix~\ref{app:paired-response} gives
\[
\begin{aligned}
\kappa_{h-k}&=\frac{11x+16}{13x+20},&
\kappa_k&=\frac{16x+20}{29x+28}.
\end{aligned}
\]

\emph{Ordering of the thresholds.} The split thresholds lie in $(1/2,2/3)$ and the pooled thresholds in $(2/3,1)$, and adjacent thresholds in each range are strictly increasing, including across $d=3$ and $d=1$; Appendix~\ref{app:paired-response} gives the differences explicitly. Together with $\kappa_0=1/2$ and $\kappa_{h+1}=1$, this proves the ordering $\kappa_0<\kappa_1<\cdots<\kappa_h<\kappa_{h+1}$ stated in the lemma.

\emph{The paired map.} Since every utility lies strictly between two consecutive thresholds, each leading coefficient is nonzero, and responses are strict. A split voter then keeps $A$ precisely when $\tau_i<k$, and a pooled voter chooses $A$ precisely when $\tau_i<h-k$. Exchanging $A$ and $B$ gives the same decision for her partner, so pairs remain split or pooled, and \eqref{eq:paired-map} follows.

\emph{The added voters.} An added voter with $C\succ A\succ B$ has private counts $(x,x,x+3d-1)$ for every $x$, including $x=h$. Her margin for $C$ over $A$ is $(1-2\gamma_j)/3>0$ at $d=0$, and the same flow sums give positive leading coefficients for every $d>0$ when $\gamma_j<1/2$ (Appendix~\ref{app:paired-response}). By Lemma~\ref{lem:top-bottom-dominance}, she never votes $B$, so every added voter stays at $C$, and all responses are strict.
\end{proof}

Lemma~\ref{lem:paired-response} reduces Theorem~\ref{thm:long-period-family} to iterating the map~\eqref{eq:paired-map} on a well-chosen set of types. 
In short, an odd number of split pairs flips every pair, while an even number applies two ordered cutoffs on the types. The construction duplicates the higher types, so that each even step changes the number of split pairs by one layer of two pairs at a time; this lengthens the poll's cycle. Two exceptional pairs, of types $s-1$ and $s$, exchange roles halfway through, which doubles the period of the complete profile.

\begin{proof}[Proof of Theorem~\ref{thm:long-period-family}]
Set $h=\ell+4=2s+1$, so that $\ell=2s-3$. The construction uses two \emph{exceptional} types, $t_1=s-1$ and $t_2=s$. We index them by $\sigma\in\{1,2\}$ and write $\bar\sigma=3-\sigma$ for the other index, so that $t_{\bar\sigma}$ is the exceptional type other than $t_\sigma$. Take the pair set
\[
P=\{\pi_0,\pi_{t_1},\pi_{t_2},\pi_{s+2},\pi_h\}
 \cup\!\bigcup_{r=s+3}^{2s}\!\{\pi_r^1,\pi_r^2\},
\]
where each pair's type is its displayed subscript, and the two copies $\pi_r^1,\pi_r^2$ both have type $r$. There are $5+2(s-2)=h$ pairs. The product of the corresponding open utility intervals is nonempty and open, all comparisons are strict, and for each resulting election the responses persist for all sufficiently small $p>0$.

For $\sigma\in\{1,2\}$ and $0\leq j\leq s-3$, define the \emph{turning} set $T_\sigma:=\{\pi_0,\pi_{t_\sigma}\}$, the \emph{odd} sets $O_{\sigma,j}:=T_\sigma\cup\{\pi_{s+2}\}\cup\bigcup_{r=s+3}^{2s-1-j}\{\pi_r^1,\pi_r^2\}$, and the \emph{even} sets $E_{\sigma,j}:=P\setminus O_{\bar\sigma,j}$, of respective sizes $2$, $2s-3-2j$ and $4+2j$. Lemma~\ref{lem:paired-response} gives the chain
\[
\begin{split}
T_\sigma\mapsto O_{\bar\sigma,0}\mapsto E_{\sigma,0}\mapsto O_{\bar\sigma,1}\mapsto E_{\sigma,1}\mapsto\cdots\\
\mapsto O_{\bar\sigma,s-3}\mapsto E_{\sigma,s-3}\mapsto T_{\bar\sigma},
\end{split}
\]
hence the recurrent orbit runs from $T_1$ to $T_2$ along the chain with $\sigma=1$, and back to $T_1$ along the chain with $\sigma=2$. 
Indeed, at $T_\sigma$ we have $k=1$: only $\pi_0$ stays split, and precisely the pooled pairs of type below $2s$ join it. Each $O_{\bar\sigma,j}$ has odd size and is therefore complemented. At $E_{\sigma,j}$, we have $k=j+2$: no split type is below $k$, while the pooled types strictly below $2s-j-1$ form $O_{\bar\sigma,j+1}$ when $j<s-3$, and $T_{\bar\sigma}$ when $j=s-3$.

The passage from $T_1$ to $T_2$ takes $1+2(s-2)=\ell$ updates, during which the numbers of split pairs are $|T_1|=2$, followed by $|O_{2,j}|=2s-3-2j$ and $|E_{1,j}|=4+2j$ for $j=0,\ldots,s-3$.
These $\ell$ numbers are pairwise distinct: after the initial $2$, the odd values decrease from $2s-3$ to $3$, while the even values increase from $4$ to $2s-2$.
Since the poll~\eqref{eq:paired-counts} is determined by the number of split pairs, it first returns after exactly $\ell$ updates. At that time, $T_1$ has become the distinct set $T_2$, and only after $\ell$ further updates do the split pairs return as well. Any return of the profile is also a return of the poll, so the profile has period exactly $2\ell$.

Finally, from truthful votes every pair is split, and Equation~\eqref{eq:paired-map} gives the entry path $P\mapsto\varnothing\mapsto P\setminus\{\pi_h\}\mapsto T_1$: the odd-size set $P$ is complemented, every pair except $\pi_h$ then splits, and at $k=s$ the cutoff retains exactly $\pi_0$ and $\pi_{t_1}$. The transient numbers of split pairs, $h$, $0$ and $2s$, lie outside the range $[2,2s-2]$ of the orbit, so the orbit is entered after exactly three updates. Since $2\leq x<h$ on the orbit, Equation~\eqref{eq:paired-counts} gives every candidate positive support and $n_C-n_A=n_C-n_B=3(h-x)>0$. For $s=3$, the chain reduces to $j=0$, which covers $\ell=3$.
\end{proof}

The theorem gives an open set of cardinal utilities, not just of rankings, since normalizing by the gap between top and bottom utilities is continuous on each strict-ranking region. Both periods grow linearly with the electorate, while the number of candidates stays at three: $C$ remains uniquely ahead throughout, yet neither the ballots nor the poll ever stabilize.

\section{Discussion}\label{sec:discussion}
\balance

\paragraph{Summary of findings.}
We studied what happens when every voter in a fixed, finite electorate responds to the same poll, using trembling-hand beliefs to decide between ballots that the poll alone leaves equivalent. A single revision typically produces a stable sincere runoff between the two poll leaders: although the dynamics are deterministic, a population drawn from any fixed distribution of preferences rarely contains the near-ties needed to prevent Duverger's law. Cycles are therefore unlikely, but their behavior can be rich: already with three candidates, cycles can have arbitrarily long periods, keep changing the poll, or let different candidates take turns as the unique winner.

\paragraph{What drives convergence, and what breaks it.}
Three ingredients shape the dynamics. 
First, separation: when the runner-up is clearly ahead, a few trembles can only make the leaders' race close, and every voter simply picks between them. 
Second, simultaneity: when all voters react to the same poll, sensible individual responses can cancel out or overshoot, and several of our cycles rely on this synchronization; whether cycles also arise under sequential updates, as in related models~\cite{meir2017acyclic,kloeckner2022cycles}, is left open. 
Third, own-ballot discounting: the same poll can prompt opposite responses from identical voters, a feature our long cycles rely on and that poll-based responses alone cannot capture.

\paragraph{Polls in practice.}
With due caution, these results suggest that a clear first poll is the main driver of coordination, while repeated polls in close three-way races may induce oscillations rather than settle them. Experiments show that real voters do react to poll information~\cite{fairstein2019polls}, and synchronous cycles in related models persist when only a large fraction of voters update~\cite{kloeckner2022cycles}. Whether such oscillations are observable in actual elections is beyond the scope of this paper.

\paragraph{Open questions.}
A first set of questions concerns where cycles arise. Our conditions for one-step convergence are sufficient, and matching necessary conditions would show where cycles can start, although a full characterization may be out of reach: even deciding whether a trembling-hand equilibrium with a given winner exists is NP-hard~\cite{obraztsova2016trembling}. Within the cycles themselves, our computer search suggests that six active voters are the minimum for period three, which remains to be proved. Whether long cycles can also change the leader, and how cycle lengths are distributed under random preferences, remain similarly open. 
A second set of questions concerns the model itself. Partial or asynchronous revisions, and other voting rules such as runoff voting, would test which of our findings depend on the synchronous plurality setting. 
More fundamentally, two stochastic variants would make the model more realistic: voters whose hands actually tremble, so that their beliefs about others' deviations become consistent with their own behavior; or polls that sample only a random subset of the electorate, so that uncertainty stems from the information rather than from beliefs about other voters. Both turn the deterministic dynamics into a Markov chain, whose long-run behavior we expect to favor Duvergerian profiles, but which calls for different tools to analyze.

\section*{CRediT Author Statement}

\textbf{Maximilian Weigel:} Methodology, Software, Formal Analysis, Investigation, Writing -- Original Draft, Visualization.
\textbf{Damien Berriaud:} Conceptualization, Methodology, Writing -- Review \& Editing, Visualization, Supervision.
\textbf{Roger Wattenhofer:} Supervision, Funding Acquisition.

\clearpage

\bibliographystyle{ACM-Reference-Format}
\bibliography{references}

@inproceedings{obraztsova2016trembling,
  author    = {Obraztsova, Svetlana and Rabinovich, Zinovi and Elkind, Edith and Polukarov, Maria and Jennings, Nicholas R.},
  title     = {Trembling Hand Equilibria of Plurality Voting},
  booktitle = {Proceedings of the Twenty-Fifth International Joint Conference on Artificial Intelligence},
  publisher = {AAAI Press},
  address = {New York, NY},
  pages     = {440--446},
  year      = {2016},
  url       = {https://www.ijcai.org/Proceedings/16/Papers/069.pdf}
}

@book{duverger1954political,
  author    = {Duverger, Maurice},
  title     = {Political Parties: Their Organization and Activity in the Modern State},
  publisher = {Methuen},
  address   = {London},
  year      = {1954}
}

@book{cox1997making,
  author    = {Cox, Gary W.},
  title     = {Making Votes Count: Strategic Coordination in the World's Electoral Systems},
  publisher = {Cambridge University Press},
  year      = {1997},
  address   = {Cambridge, UK},
  doi       = {10.1017/CBO9781139174954}
}

@inproceedings{reijngoud2012polls,
  author    = {Reijngoud, Annemieke and Endriss, Ulle},
  title     = {Voter Response to Iterated Poll Information},
  booktitle = {Proceedings of the 11th International Conference on Autonomous Agents and Multiagent Systems},
  publisher = {International Foundation for Autonomous Agents and Multiagent Systems},
  address   = {Richland, SC},
  pages     = {635--644},
  year      = {2012}
}

@inproceedings{meir2014local,
  author    = {Meir, Reshef and Lev, Omer and Rosenschein, Jeffrey S.},
  title     = {A Local-Dominance Theory of Voting Equilibria},
  booktitle = {Proceedings of the Fifteenth ACM Conference on Economics and Computation},
  publisher = {ACM},
  address = {New York, NY},
  pages     = {313--330},
  year      = {2014},
  doi       = {10.1145/2600057.2602860}
}

@inproceedings{meir2015uncertainty,
  author    = {Meir, Reshef},
  title     = {Plurality Voting under Uncertainty},
  booktitle = {Proceedings of the Twenty-Ninth AAAI Conference on Artificial Intelligence},
  publisher = {AAAI Press},
  address = {Palo Alto, CA},
  pages     = {2103--2109},
  year      = {2015},
  doi       = {10.1609/aaai.v29i1.9435}
}

@article{meir2017acyclic,
  author  = {Meir, Reshef and Polukarov, Maria and Rosenschein, Jeffrey S. and Jennings, Nicholas R.},
  title   = {Iterative Voting and Acyclic Games},
  journal = {Artificial Intelligence},
  volume  = {252},
  pages   = {100--122},
  year    = {2017},
  doi     = {10.1016/j.artint.2017.08.002}
}

@book{meir2018strategic,
  author    = {Meir, Reshef},
  title     = {Strategic Voting},
  series    = {Synthesis Lectures on Artificial Intelligence and Machine Learning},
  publisher = {Morgan \& Claypool Publishers},
  address   = {San Rafael, CA},
  year      = {2018},
  doi       = {10.2200/S00849ED1V01Y201804AIM038}
}

@inproceedings{vasselai2022icv,
  author    = {Vasselai, Fabricio},
  title     = {Iterative Calculus of Voting under Plurality},
  booktitle = {Proceedings of the Thirty-Sixth AAAI Conference on Artificial Intelligence},
  publisher = {AAAI Press},
  address = {Palo Alto, CA},
  volume    = {36},
  pages     = {5208--5218},
  year      = {2022},
  doi       = {10.1609/aaai.v36i5.20456}
}

@article{kloeckner2022cycles,
  author  = {Kloeckner, Beno\^it R.},
  title   = {Cycles in Synchronous Iterative Voting: General Robustness and Examples in Approval Voting},
  journal = {Social Choice and Welfare},
  volume  = {59},
  number  = {2},
  pages   = {423--466},
  year    = {2022},
  doi     = {10.1007/s00355-022-01395-6}
}

@article{myerson1993voting,
  author  = {Myerson, Roger B. and Weber, Robert J.},
  title   = {A Theory of Voting Equilibria},
  journal = {American Political Science Review},
  volume  = {87},
  number  = {1},
  pages   = {102--114},
  year    = {1993},
  doi     = {10.2307/2938959}
}

@inproceedings{meir2010convergence,
  author    = {Meir, Reshef and Polukarov, Maria and Rosenschein, Jeffrey S. and Jennings, Nicholas R.},
  title     = {Convergence to Equilibria in Plurality Voting},
  booktitle = {Proceedings of the Twenty-Fourth AAAI Conference on Artificial Intelligence},
  publisher = {AAAI Press},
  address = {Palo Alto, CA},
  pages     = {823--828},
  year      = {2010},
  doi       = {10.1609/aaai.v24i1.7624}
}

@techreport{messner2011miscounts,
  author      = {Messner, Matthias and Polborn, Mattias K.},
  title       = {Miscounts, {Duverger}'s Law and {Duverger}'s Hypothesis},
  institution = {IGIER, Bocconi University},
  type        = {Working Paper},
  number      = {380},
  year        = {2011},
  url         = {https://igier.unibocconi.eu/sites/default/files/media/publication/380.pdf}
}

@inproceedings{goldberg2026complexity,
  author    = {Goldberg, Paul W. and Mavronicolas, Marios and W{\k{a}}s, Tomasz},
  title     = {Complexity of (Non-)Convergence in Iterative Voting},
  booktitle = {Proceedings of the 25th International Conference on Autonomous Agents and Multiagent Systems},
  publisher = {International Foundation for Autonomous Agents and Multiagent Systems},
  address   = {Richland, SC},
  pages     = {2545--2553},
  year      = {2026},
  doi       = {10.65109/MDLZ8000},
  url       = {https://ifaamas.org/Proceedings/aamas2026/pdfs/MDLZ8000.pdf}
}

@incollection{palfrey1989duverger,
  author    = {Palfrey, Thomas R.},
  title     = {A Mathematical Proof of Duverger's Law},
  booktitle = {Models of Strategic Choice in Politics},
  editor    = {Ordeshook, Peter C.},
  publisher = {University of Michigan Press},
  address   = {Ann Arbor},
  pages     = {69--91},
  year      = {1989}
}

@article{fey1997stability,
  author  = {Fey, Mark},
  title   = {Stability and Coordination in Duverger's Law: A Formal Model of Preelection Polls and Strategic Voting},
  journal = {American Political Science Review},
  volume  = {91},
  number  = {1},
  pages   = {135--147},
  year    = {1997},
  doi     = {10.2307/2952264}
}

@inproceedings{obraztsova_plurality_2013,
	address = {Berlin, Heidelberg},
	title = {Plurality {Voting} with {Truth}-{Biased} {Agents}},
	isbn = {978-3-642-41392-6},
	doi = {10.1007/978-3-642-41392-6_3},
	language = {en},
	booktitle = {Algorithmic {Game} {Theory}},
	publisher = {Springer},
	author = {Obraztsova, Svetlana and Markakis, Evangelos and Thompson, David R. M.},
	editor = {Vöcking, Berthold},
	year = {2013},
	pages = {26--37},
}

@article{forand_dynamic_2015,
	title = {A dynamic {Duverger}’s law},
	volume = {165},
	issn = {1573-7101},
	url = {https://doi.org/10.1007/s11127-016-0309-8},
	doi = {10.1007/s11127-016-0309-8},
	language = {en},
	number = {3},
	urldate = {2026-10-08},
	journal = {Public Choice},
	author = {Forand, Jean Guillaume and Maheshri, Vikram},
	month = dec,
	year = {2015},
	pages = {285--306},
}

@inproceedings{desmedt_equilibria_2010,
	address = {New York, NY, USA},
	series = {{EC} '10},
	title = {Equilibria of plurality voting with abstentions},
	isbn = {978-1-60558-822-3},
	url = {https://dl.acm.org/doi/10.1145/1807342.1807398},
	doi = {10.1145/1807342.1807398},
	urldate = {2026-10-08},
	booktitle = {Proceedings of the 11th {ACM} conference on {Electronic} commerce},
	publisher = {Association for Computing Machinery},
	author = {Desmedt, Yvo and Elkind, Edith},
	month = jun,
	year = {2010},
	pages = {347--356},
}

@inproceedings{lev_convergence_2012,
	address = {Richland, SC},
	series = {{AAMAS} '12},
	title = {Convergence of iterative voting},
	isbn = {978-0-9817381-2-3},
	url = {https://dl.acm.org/doi/10.5555/2343776.2343784},
	urldate = {2026-10-08},
	booktitle = {Proceedings of the 11th {International} {Conference} on {Autonomous} {Agents} and {Multiagent} {Systems} - {Volume} 2},
	publisher = {International Foundation for Autonomous Agents and Multiagent Systems},
	author = {Lev, Omer and Rosenschein, Jeffrey S.},
	month = jun,
	year = {2012},
	pages = {611--618},
}

@inproceedings{rabinovich_analysis_2015,
	title = {Analysis of {Equilibria} in {Iterative} {Voting} {Schemes}},
	volume = {29},
    pages = {1007--1013},
	copyright = {Copyright (c)},
	issn = {2374-3468},
	url = {https://ojs.aaai.org/index.php/AAAI/article/view/9328},
	doi = {10.1609/aaai.v29i1.9328},
	language = {en},
	urldate = {2026-10-08},
	booktitle = {Proceedings of the AAAI Conference on Artificial Intelligence},
	author = {Rabinovich, Zinovi and Obraztsova, Svetlana and Lev, Omer and Markakis, Evangelos and Rosenschein, Jeffrey},
    publisher = {AAAI Press},
	address = {Palo Alto, CA},
	month = feb,
	year = {2015},
}

@inproceedings{branzei_how_2013,
	address = {Bellevue, Washington},
	series = {{AAAI}'13},
	title = {How bad is selfish voting?},
	urldate = {2026-10-08},
	booktitle = {Proceedings of the {Twenty}-{Seventh} {AAAI} {Conference} on {Artificial} {Intelligence}},
	publisher = {AAAI Press},
	author = {Brânzei, Simina and Caragiannis, Ioannis and Morgenstern, Jamie and Procaccia, Ariel D.},
	month = jul,
	year = {2013},
	pages = {138--144},
}

@inproceedings{thompson_empirical_2013,
	address = {Richland, SC},
	series = {{AAMAS} '13},
	title = {Empirical analysis of plurality election equilibria},
	isbn = {978-1-4503-1993-5},
	url = {https://dl.acm.org/doi/10.5555/2484920.2484983},
	urldate = {2026-10-08},
	booktitle = {Proceedings of the 2013 international conference on {Autonomous} agents and multi-agent systems},
	publisher = {International Foundation for Autonomous Agents and Multiagent Systems},
	author = {Thompson, David R.M. and Lev, Omer and Leyton-Brown, Kevin and Rosenschein, Jeffrey},
	month = may,
	year = {2013},
	pages = {391--398},
}

@inproceedings{kavner_strategic_2021,
	address = {Red Hook, NY, USA},
	series = {{NIPS} '21},
	title = {Strategic behavior is bliss: iterative voting improves social welfare},
	isbn = {978-1-7138-4539-3},
	shorttitle = {Strategic behavior is bliss},
	urldate = {2026-10-08},
	booktitle = {Proceedings of the 35th {International} {Conference} on {Neural} {Information} {Processing} {Systems}},
	publisher = {Curran Associates Inc.},
	author = {Kavner, Joshua and Xia, Lirong},
	month = dec,
	year = {2021},
	pages = {19021--19032},
}

@inproceedings{wilczynski2019poll,
  author    = {Wilczynski, Ana{\"e}lle},
  title     = {Poll-Confident Voters in Iterative Voting},
  booktitle = {Proceedings of the Thirty-Third AAAI Conference on Artificial Intelligence},
  publisher = {AAAI Press},
  address = {Palo Alto, CA},
  volume    = {33},
  pages     = {2205--2212},
  year      = {2019},
  doi       = {10.1609/aaai.v33i01.33012205}
}

@inproceedings{fairstein2019polls,
  author    = {Fairstein, Roy and Lauz, Adam and Meir, Reshef and Gal, Kobi},
  title     = {Modeling People's Voting Behavior with Poll Information},
  booktitle = {Proceedings of the 18th International Conference on Autonomous Agents and MultiAgent Systems},
  publisher = {International Foundation for Autonomous Agents and Multiagent Systems},
  address   = {Richland, SC},
  pages     = {1422--1430},
  year      = {2019},
  url       = {https://www.ifaamas.org/Proceedings/aamas2019/pdfs/p1422.pdf}
}

@misc{kavner2024average,
  author        = {Kavner, Joshua and Xia, Lirong},
  title         = {Average-Case Analysis of Iterative Voting},
  year          = {2024},
  eprint        = {2402.08144},
  archivePrefix = {arXiv},
  primaryClass  = {cs.GT},
  url           = {https://arxiv.org/abs/2402.08144}
}

@article{selten1975perfectness,
  author  = {Selten, Reinhard},
  title   = {Reexamination of the Perfectness Concept for Equilibrium Points in Extensive Games},
  journal = {International Journal of Game Theory},
  volume  = {4},
  number  = {1},
  pages   = {25--55},
  year    = {1975},
  doi     = {10.1007/BF01766400}
}

@article{dhillon_when_2004,
	title = {When are plurality rule voting games dominance-solvable?},
	volume = {46},
	issn = {0899-8256},
	url = {https://www.sciencedirect.com/science/article/pii/S0899825603000502},
	doi = {10.1016/S0899-8256(03)00050-2},
	number = {1},
	urldate = {2026-10-09},
	journal = {Games and Economic Behavior},
	author = {Dhillon, Amrita and Lockwood, Ben},
	month = jan,
	year = {2004},
	pages = {55--75},
}


\clearpage
\appendix

\numberwithin{equation}{section}

This Appendix gives the coefficient calculations and case reductions used in the paper. All calculations use three candidates and uniform trembles: each other ballot stays with probability $1-p$ and moves to either alternative with probability $p/2$. As in Section~\ref{sec:model}, a tied winner set is valued by its uniform lottery, $q$ denotes the private counts, which exclude the focal ballot, and $v$ denotes a terminal private vector.

For a fixed voter, we suppress her index and write $\Delta^p(c,c'\mid q)=U^p(c\mid q)-U^p(c'\mid q)$. A margin written $c-c'$ means this difference at the displayed private vector. We use $A,B,C$ for named candidates and $X\succ Y\succ Z$ for a voter's favorite, middle, and least-preferred candidate. The first nonzero coefficient determines a margin's sign for sufficiently small positive $p$; if a coefficient vanishes, later orders decide.

\section{Counting Directed Trembles}
We first give a formula for calculating expected-utility differences by group the believed trembles according to their number. The first nonzero coefficient determines the limiting comparison as $p\to0$. Nevertheless, we have to take all contributions at that order into account. Lemma~\ref{lem:directed-flow-expansion} provides the formula used in the response calculations and cycle examples below.

\begin{lemma}
\label{lem:directed-flow-expansion}
Fix a voter with private counts $q=(q_A,q_B,q_C)$ and write $|q|:=q_A+q_B+q_C$. For $\eta\in\mathbb N_0$, let $\mathcal F_\eta(q)$ be the set of flows $f=(f_{cc'})_{c\neq c'}$ of nonnegative integers satisfying
\[
    \sum_{c\neq c'}f_{cc'}=\eta,
    \qquad
    \sum_{c'\neq c}f_{cc'}\leq q_c
    \quad\text{for every }c.
\]
Here $f_{cc'}$ counts the ballots moving from $c$ to $c'$, and the inequalities ensure that the flows do not exceed the number of counts they depart from. The change in the private counts is
\[
    (Df)_c:=\sum_{c'\neq c}f_{c'c}-\sum_{c'\neq c}f_{cc'}
\]
so the resulting private counts are $q+Df$. The number of assignments of individual ballots producing this flow is
\[
    M_q(f)
    :=
    \prod_c
    \frac{q_c!}
    {\bigl(q_c-\sum_{c'\neq c}f_{cc'}\bigr)!
     \prod_{c'\neq c}f_{cc'}!}.
\]
For two contemplated ballots $c$ and $c'$, define
\begin{align*}
 H_\eta(c,c';q)
 :=
 \sum_{f\in\mathcal F_\eta(q)}
 M_q(f)
 \big[
 &u\bigl(W(q+Df+\mathbf e_c)\bigr)\\
 &-u\bigl(W(q+Df+\mathbf e_{c'})\bigr)
 \big].
\end{align*}

Under the uniform three-candidate tremble belief,
\[
 U^p(c\mid q)-U^p(c'\mid q)
 =
 \sum_{\eta=0}^{|q|}
 \left(\frac p2\right)^\eta
 (1-p)^{|q|-\eta}
 H_\eta(c,c';q).
\]
In particular, if $d\in\mathbb N_0$ and $H_\eta(c,c';q)=0$ for every $\eta<d$, then
\[
 [p^d]\bigl(U^p(c\mid q)-U^p(c'\mid q)\bigr)
 =2^{-d}H_d(c,c';q).
\]
\end{lemma}

\begin{proof}
For each candidate $c$, the multinomial factor in $M_q(f)$ chooses which of the $q_c$ ballots supporting $c$ move to either other candidate. Each moved ballot contributes $p/2$ and each unchanged ballot contributes $1-p$. A flow with $\eta$ trembles therefore has probability
\[
 M_q(f)\left(\frac p2\right)^\eta(1-p)^{|q|-\eta}.
\]
Summing its payoff difference over all flows proves the expansion. If $H_\eta=0$ for every $\eta<d$, the terms with $\eta<d$ vanish and those with $\eta>d$ contain the factor $p^\eta$. The coefficient of $p^d$ is therefore $2^{-d}H_d$.
\end{proof}

\section{Proof of Theorem~\ref{thm:condensed-one-step}}

Theorem~\ref{thm:condensed-one-step} asserts both that the first update concentrates support on the initial leaders and that the resulting profile is stable. We first calculate the individual responses under the given gap conditions, where we separate the even and odd leader gaps. Then we verify that no voter switches at the resulting profile.

\subsection{The Three-Candidate Response Calculation}
\label{app:one-step-blocks}

Write
\[
 (n_A^0,n_B^0,n_C^0)=(r+d,r,r-s),
 \qquad r\geq2,\quad d\geq0,\quad 1\leq s\leq r-1.
\]
The private vectors of an $A$-, $B$-, and $C$-voter are
\begin{equation}
 q^{A}=(r+d-1,r,r-s),\quad
 q^{B}=(r+d,r-1,r-s),\quad
 q^{C}=(r+d,r,r-s-1).
\label{eq:app-B1}
\end{equation}
For terminal private counts $v$, one directed tremble changes the two gaps as follows (the rows always give the decrease):
\begin{equation}
\begin{array}{c|rrrrrr}
 &A\mapsto B&A\mapsto C&B\mapsto A&B\mapsto C&C\mapsto A&C\mapsto B\\ \hline
v_A-v_B&2&1&-2&-1&-1&1\\
v_A-v_C&1&2&-1&1&-2&-1
\end{array}
\label{eq:app-B2}
\end{equation}
Each tremble reduces a given gap by at most two. Reducing it by at least $2j$ in $j$ trembles therefore requires every tremble to move from the leader to the other candidate. To reduce it by at least $2j-1$, either all $j$ trembles reduce it by two or exactly one reduces it by one. Checking when the contemplated ballots give different winner sets gives the complete flow lists below and shows that all earlier coefficients vanish.

For $d=0$, staying restores the $A$--$B$ tie for a truthful $A$-voter. Switching to $B$ makes $B$ win alone; switching to $C$ leaves $B$ alone when $s\geq2$ and creates a $B$--$C$ tie when $s=1$. Both alternatives are strictly worse because $A$ is her favorite. The symmetric comparison holds for a truthful $B$-voter. A truthful $C$-voter chooses between making $A$ or $B$ win alone and retaining the $A$--$B$ tie. She therefore chooses the one of $A,B$ she prefers. Suppose first that $d=2k>0$ and put $a=r+2k$, with $\binom{x}{-1}=0$. Lemma~\ref{lem:directed-flow-expansion} gives the following first nonzero coefficients:
\begingroup
\scriptsize
\begin{equation}
\begin{array}{c|c|c|l}
q&\text{Margin}&\text{Order}&\text{Coefficient}\\ \hline
q^{A}&A-B&k-1&2^{-k}\binom{a-1}{k-1}(u(A)-u(B))\\[2pt]
q^{B}&B-A&k&2^{-k-1}\binom ak(u(B)-u(A))\\[2pt]
q^{B}&B-C&k&2^{-k-1}\binom ak
 \begin{cases}u(B)-u(C),&s=1,\\u(B)-u(A),&s\geq2,
 \end{cases}\\[7pt]
q^{C}&A-C&k&2^{-k-1}\binom ak(u(A)-u(B))\\[2pt]
q^{C}&B-C&k&2^{-k-1}
 \big((k+1)\binom ak\\&&&+(r-s-1)\binom a{k-1}\big)(u(B)-u(A)).
\end{array}
\label{eq:app-B3}
\end{equation}
\endgroup
The contributing flows, in the same row order as \eqref{eq:app-B3}, are
\begin{equation}
\begin{array}{c|l}
1&(k-1)(A\mapsto B)\\
2&k(A\mapsto B)\\
3&k(A\mapsto B),\quad k(A\mapsto C)\ \text{only when }s=1\\
4&k(A\mapsto B)\\
5&k(A\mapsto B),\quad (k-1)(A\mapsto B)+(A\mapsto C\text{ or }C\mapsto B).
\end{array}
\label{eq:app-B4}
\end{equation}
For example, in the third row with $s=1$ the two equal multiplicities $\binom ak$ multiply payoff changes $(u(B)-u(A))/2$ and $(u(A)-u(C))/2$. In the last row, the two mixed-flow multiplicities sum to $k\binom ak+(r-s-1)\binom a{k-1}$; adding the pure flow gives \eqref{eq:app-B3}.

The remaining even-gap comparison is $A-C$ at $q^{A}$. Below order $k$ both ballots leave $A$ the unique winner. At order $k$, \eqref{eq:app-B2} implies $v_A\geq v_C$ and $v_B\leq v_A+1$. If $v_B\leq v_A$, ballot $A$ makes the voter's favorite the unique winner. Otherwise $v_B=v_A+1$: ballot $A$ gives an $A$--$B$ tie, whereas ballot $C$ gives either $B$ alone or a $B$--$C$ tie. Thus every contribution is nonnegative for $A\succ C\succ B$. The flow $k(A\mapsto B)$ is strictly positive, with payoff change $(u(A)-u(B))/2$. Therefore its order-$k$ coefficient is strictly positive.

Now let $d=2k+1$ and $s\geq2$, and put $a=r+2k+1$. The first nonzero coefficients are
\begingroup
\scriptsize
\begin{equation}
\begin{array}{c|c|c|l}
q&\text{Margin}&\text{Order}&\text{Coefficient}\\ \hline
q^{A}&A-B&k&2^{-k}\left(\binom{a-1}k+\frac{2r+k+1-s}{2}\binom{a-1}{k-1}\right)(u(A)-u(B))\\[3pt]
q^{A}&A-C&k&2^{-k-1}\binom{a-1}k(u(A)-u(B))\\[3pt]
q^{B}&B-A&k+1&2^{-k-1}\left(\binom a{k+1}+\frac{2r+k+1-s}{2}\binom ak\right)(u(B)-u(A))\\[3pt]
q^{B}&B-C&k+1&2^{-k-2}\big(
 [\binom a{k+1}+(2r+k+1-s)\binom ak]\\&&&\cdot(u(B)-u(A))
 +1_{\{s=2\}}\binom a{k+1}(u(A)-u(C))\big)\\[3pt]
q^{C}&A-C&k+1&2^{-k-2}\left(\binom a{k+1}+(2r+k-s)\binom ak\right)(u(A)-u(B))\\[3pt]
q^{C}&B-C&k&2^{-k-1}\binom ak(u(B)-u(A)).
\end{array}
\label{eq:app-B5}
\end{equation}
\endgroup
The decisive flows in the six rows are, respectively,
\begin{equation}
\begin{array}{c|l}
1&k(A\mapsto B),\ (k-1)(A\mapsto B)+(A\mapsto C\text{ or }C\mapsto B)\\
2&k(A\mapsto B)\\
3&(k+1)(A\mapsto B),\ k(A\mapsto B)+(A\mapsto C\text{ or }C\mapsto B)\\
4&\text{the flows in row 3; also }(k+1)(A\mapsto C)\ \text{when }s=2\\
5&(k+1)(A\mapsto B),\ k(A\mapsto B)+(A\mapsto C\text{ or }C\mapsto B)\\
6&k(A\mapsto B).
\end{array}
\label{eq:app-B6}
\end{equation}
Terms containing $k-1$ are absent when $k=0$. The mixed-flow multiplicities in rows 1, 3, and 5 are $(2r+k+1-s)\binom{a-1}{k-1}$, $(2r+k+1-s)\binom ak$, and $(2r+k-s)\binom ak$, respectively. Formula \eqref{eq:app-B2} also shows that $C$ cannot enter an earlier winner set and can enter the first contributing order only when $s=2$. In that case, the fourth coefficient in \eqref{eq:app-B5} rewrites as
\begin{equation}
 2^{-k-2}\left[
 \binom a{k+1}(u(B)-u(C))
 +(2r+k-1)\binom ak(u(B)-u(A))
 \right]>0.
\label{eq:app-B7}
\end{equation}
When $k=0=s-2$, the flow $A\mapsto C$ in \eqref{eq:app-B6} is one flow whose payoff is corrected from $(u(B)-u(A))/2$ to $(u(B)-u(C))/2$, not a second copy. Equations \eqref{eq:app-B3}--\eqref{eq:app-B7}, together with Lemma~\ref{lem:top-bottom-dominance}, prove every entry of the response table in the main proof.

\subsection{Why the Successor Is Stable}
\label{app:one-step-fixedness}

Removing all support from $C$ after one update does not by itself establish convergence. We must also show that every voter keeps her new ballot. The argument below covers the resulting profile for the remaining two candidates.

Suppose a profile has counts $(a,b,0)$ with $a,b\geq3$, every $A$-supporter prefers $A$ to $B$, and every $B$-supporter prefers $B$ to $A$. An $A$-supporter has private counts $q=(a-1,b,0)$. If $q_A\geq3$, Lemma~\ref{lem:general-two-candidates} makes her stay. Otherwise $q_A=2$. Put $\eta=\lfloor(q_B-2)/2\rfloor$. Through at most $\eta$ trembles, every terminal private vector $v$ satisfies $v_B-v_C\geq q_B-2\eta\geq2$, so $C$ cannot win after any contemplated ballot. Voting for $A$ is therefore weakly better than either alternative because $u(A)>u(B)$. The flow $\eta(B\mapsto A)$ leaves $v_B-v_A\in\{0,1\}$ and makes ballot $A$ strictly better than both alternatives. For each comparison, its first nonzero coefficient is consequently a positive sum of contributions at some order at most $\eta$, by Lemma~\ref{lem:directed-flow-expansion}. Exchanging $A,B$ proves that every supporter stays, so the profile is stable.

The successor in Theorem~\ref{thm:condensed-one-step} is
\begin{equation}
 (r+d+x,\,2r-s-x,\,0),\qquad 0\leq x\leq r-s,
\label{eq:app-B8}
\end{equation}
and every supporter prefers her current to the other supported candidate. Both supported counts are at least three when $r\geq3$, so the preceding argument applies. If $r=2$, then $s=1$, $d=2j-2$ for some $j\geq1$, and the two possible successors are $(2j,3,0)$ and $(2j+1,2,0)$. The first is covered above for $j\geq2$. For $(2,3,0)$, the required margins are
\begin{equation}
\begin{array}{c|c|c}
q&\text{Margin}&\text{First coefficient}\\ \hline
(1,3,0)&A-B&[p]\Delta^p=\frac94(u(A)-u(B))\\
         &A-C&[p]\Delta^p=\frac34(2u(A)-u(B)-u(C))\\
(2,2,0)&B-A&[p^0]\Delta^p=u(B)-u(A)\\
         &B-C&[p^0]\Delta^p=(u(B)-u(A))/2.
\end{array}
\label{eq:app-B9}
\end{equation}
In $(2,3,0)$, both $A$-supporters are original truthful $A$-voters, so $2u(A)-u(B)-u(C)>0$ in \eqref{eq:app-B9}. For $(2j+1,2,0)$, an $A$-supporter has $q=(2j,2,0)$ and a $B$-supporter has $q=(2j+1,1,0)$. With $\mu_j=2^{-j-1}\binom{2j+1}j$, the first coefficients are
\begingroup
\small
\begin{equation}
\begin{aligned}
\relax[p^{j-1}]\Delta^p(A,C\mid(2j,2,0))
 &=2^{-j}\binom{2j}{j-1}(u(A)-u(B)),\\
\relax[p^{j-1}]\Delta^p(A,B\mid(2j,2,0))
 &=2^{-j}(j+1)\binom{2j}{j-1}(u(A)-u(B)),\\
\relax[p^j]\Delta^p(B,A\mid(2j+1,1,0))
 &=\mu_j(j+2)(u(B)-u(A)),\\
\relax[p^j]\Delta^p(B,C\mid(2j+1,1,0))
 &=\mu_j\bigl((j+1)u(B)-ju(A)-u(C)\bigr).
\end{aligned}
\label{eq:app-B10}
\end{equation}
\endgroup
The first two lines come from the pure flow $(j-1)(A\mapsto B)$ and, for $A-B$, the additional flow $(j-2)(A\mapsto B)+(A\mapsto C)$, which is absent when $j=1$. The last two come from $j(A\mapsto B)$ and $(j-1)(A\mapsto B)+(A\mapsto C)$; for $B-C$ one also includes $j(A\mapsto C)$. At $j=1$, the last two descriptions of $A\mapsto C$ denote the same flow, which is counted once with its actual payoff. All four expressions are positive: the last equals $\mu_j[j(u(B)-u(A))+u(B)-u(C)]$, and every $B$-supporter in this case is a truthful $B$-voter. Hence both boundary successors are stable, completing the calculation.
\section{Proof of Theorem~\ref{thm:no-three-active-period-three}}
\label{supp:active-proof}

\begin{reptheorem}{thm:no-three-active-period-three}
With three candidates, every period-three cycle of the complete ballot profile has at least four active voters.
\end{reptheorem}

We prove the lower bound in Theorem~\ref{thm:no-three-active-period-three} without assuming truthful initialization or restricting the ballots of inactive voters. We first reduce the possible ballot paths to two stay patterns, then prove the response comparisons needed to exclude both. The final two subsections verify that the case reductions used in this argument are indeed exhaustive.

\begingroup

\subsection{Ballot Paths and Stay Patterns}

\begin{lemma}
\label{lem:condensed-active-path}
\label{lem:condensed-active-reduction}
On a period-three cycle, every active voter uses her favorite and middle candidate and stays in exactly one of the three transitions. There are at least three active voters. With exactly three, the numbers of active stayers are, up to a cyclic shift,
\[
(2,1,0)\qquad\text{or}\qquad(1,1,1).
\]
\end{lemma}

\begin{proof}
Lemma~\ref{lem:top-bottom-dominance} restricts every ballot on the orbit to the voter's two highest-ranked candidates. A nonconstant cyclic word of length three on two candidates has exactly one stay. Equal private counts require equal responses. In particular, a voter who is the only mover in a transition must stay in the next transition: the other ballots have not changed between the two source phases.

One active voter faces constant private counts and cannot cycle. With two active voters, a stay by voter $1$ makes voter $2$ the sole mover, so voter $2$ stays next. Voter $1$ is then the sole mover and must stay again, contradicting her unique stay. A transition in which both stay would already be stable.

With three active voters, there are three stays in total, and no transition has three stays. If a transition has two stays, its sole mover stays next, so the stay counts are $(2,1,0)$ up to a cyclic shift. Otherwise every transition has one stay.
\end{proof}

\subsection{Response Comparisons}
\label{app:three-active}
\label{sec:three-active-response-details}

The following comparisons restrict the responses that the two stay patterns can require. We prove them for a single voter and then apply them to the private counts determined by the other active ballots and the fixed inactive ballots.

Fix one voter and rename her candidates so that
\[
X\succ Y\succ Z,
\qquad
(u(X),u(Y),u(Z))=(1,\rho,0),
\qquad 0<\rho<1.
\]
We write $\operatorname{BR}(q)$ for her response to the private count vector $q$. The directed-flow expansion of Lemma~\ref{lem:directed-flow-expansion} will be used throughout.

\begin{lemma}
\label{lem:three-active-middle-score}
If $q_X\geq q_Y$, ballot $X$ strictly beats ballot $Y$ for every $0<p<2/3$. Hence a limiting middle response requires $q_Y>q_X$.
\end{lemma}

\begin{proof}
Normalize the utilities to $(u(X),u(Y),u(Z))=(1,\rho,0)$, where $0<\rho<1$. Put $\alpha=1-p$ and $\beta=p/2$. For $0<p<2/3$, we have $\alpha>\beta>0$.

First compare the probabilities of terminal private vectors obtained by exchanging the counts of $X$ and $Y$. Let $P_q(v)$ be the probability of $v$, and let $\sigma v$ exchange its $X$- and $Y$-coordinates. We show that
\begin{equation}
P_q(v)\geq P_q(\sigma v)
\qquad\text{whenever }v_X\geq v_Y.
\label{eq:app-C1}
\end{equation}
Write $q_X=q_Y+k$, where $k\geq0$. The probability $P_q(v)$ is the coefficient of $s^{v_X}t^{v_Y}w^{v_Z}$ in
\[
\begin{aligned}
&\bigl[(\alpha s+\beta t+\beta w)
       (\beta s+\alpha t+\beta w)\bigr]^{q_Y}\\
&\qquad{}\cdot(\alpha s+\beta t+\beta w)^k
       [\beta(s+t)+\alpha w]^{q_Z}.
\end{aligned}
\]
The two relevant factors can be written as
\[
\begin{aligned}
&(\alpha s+\beta t+\beta w)(\beta s+\alpha t+\beta w)\\
&\quad=\alpha\beta(s+t)^2+(\alpha-\beta)^2st
+\beta(\alpha+\beta)w(s+t)+\beta^2w^2,
\end{aligned}
\]
and
\[
\alpha s+\beta t+\beta w
=\beta(s+t)+(\alpha-\beta)s+\beta w.
\]
All displayed coefficients are nonnegative. After fixing the power of $w$, expanding these factors therefore gives a sum of nonnegative multiples of terms
\[
s^a t^b(s+t)^\ell,
\qquad a\geq b\geq0,\quad \ell\geq0.
\]
In each such term, the coefficients corresponding to $v$ and $\sigma v$ are
\[
\binom{\ell}{v_X-a},
\qquad
\binom{\ell}{v_Y-a},
\]
respectively. We interpret a binomial coefficient as zero when its lower index is outside $\{0,\ldots,\ell\}$. Since $v_X+v_Y=a+b+\ell$, the two lower indices satisfy
\[
v_X-a\geq v_Y-a,
\qquad
(v_X-a)+(v_Y-a)=\ell-(a-b)\leq\ell.
\]
The first index is therefore at least as close to $\ell/2$ as the second. Binomial coefficients increase toward the middle and have equal values at reflected indices, so the first coefficient is at least the second. Adding these comparisons proves \eqref{eq:app-C1}.

Next compare the payoffs. Put
\[
D(v)=u(W(v+\mathbf e_X))-u(W(v+\mathbf e_Y)).
\]
For $v_X\geq v_Y$, all nonzero possibilities are listed below:
\begin{equation}
\begin{array}{c|c|c|c}
v_X-v_Y&v_Z-v_X&D(v)&D(\sigma v)\\ \hline
0&\leq0&1-\rho&1-\rho\\
0&1&(1-\rho)/2&(1-\rho)/2\\
1&\leq-1&(1-\rho)/2&(1-\rho)/2\\
1&0&(2-\rho)/3&(1-2\rho)/3\\
1&1&1/2&-\rho/2\\
\geq2&\in\{0,1\}&1/2&-\rho/2.
\end{array}
\label{eq:app-C2}
\end{equation}
Any omitted entry is zero. In every row, $D(v)\geq0$ and $D(v)+D(\sigma v)\geq0$. For $v_X>v_Y$, the combined contribution of $v$ and $\sigma v$ is therefore
\[
\begin{aligned}
&P_q(v)D(v)+P_q(\sigma v)D(\sigma v)\\
&\quad=[P_q(v)-P_q(\sigma v)]D(v)\\
&\qquad{}+P_q(\sigma v)[D(v)+D(\sigma v)]
\geq0,
\end{aligned}
\]
using \eqref{eq:app-C1}. When $v_X=v_Y$, the vector is counted once and its contribution is also nonnegative.

Finally, at least one contribution is strictly positive. If $|q|=2r$, take $v=(r,r,0)$; if $|q|=2r+1$, take $v=(r,r,1)$, in $(X,Y,Z)$-coordinates. In both cases, \eqref{eq:app-C2} gives $D(v)>0$, and full support gives $P_q(v)>0$. Thus $U^p(X\mid q)>U^p(Y\mid q)$.
\end{proof}

\begin{lemma}
\label{lem:three-active-score-obstructions}
A top response $X$ is impossible if
\[
q_Z\geq q_Y\geq q_X+2
\qquad\text{or}\qquad
\min\{q_Y,q_Z\}\geq q_X+3.
\]
\end{lemma}

\begin{proof}
We record two elementary ways in which a proposed top response can fail. If $v_X+1<\max\{v_Y,v_Z\}$, adding $X$ cannot make $X$ a winner. The comparison between ballots $X$ and $Y$ is then
\begin{equation}
\begin{array}{c|c|c|c}
v_Y-v_Z&W(v+\mathbf e_X)&W(v+\mathbf e_Y)&
u(W(v+\mathbf e_Y))-u(W(v+\mathbf e_X))\\ \hline
\leq-2&\{Z\}&\{Z\}&0\\
-1&\{Z\}&\{Y,Z\}&\rho/2\\
0&\{Y,Z\}&\{Y\}&\rho/2\\
\geq1&\{Y\}&\{Y\}&0.
\end{array}
\label{eq:app-C3}
\end{equation}
Put
\[
\eta=
\begin{cases}
\lfloor(q_Z-q_Y)/2\rfloor,&q_Z\geq q_Y,\\
\lceil(q_Y-q_Z)/2\rceil,&q_Y>q_Z,
\end{cases}
\]
and let $c=Z$ in the first case and $c=Y$ in the second. If
\begin{equation}
q_c-q_X\geq2\eta+2,
\label{eq:app-C4}
\end{equation}
then after at most $\eta$ trembles the score of $c$ remains at least two above that of $X$, so adding $X$ cannot put $X$ in the winner set. With fewer than $\eta$ trembles, adding $Y$ cannot alter the winner set either, by the choice of $\eta$. With exactly $\eta$ trembles, \eqref{eq:app-C3} weakly favors $Y$ for every flow and strictly favors it for the flow that moves $\eta$ ballots from the higher of $Y,Z$ to the lower. Hence $Y$ asymptotically strictly beats $X$. Applying \eqref{eq:app-C4}, and also observing that a $Y$--$Z$ tie favors $Y$ already at order zero, proves the statement once its two hypotheses are checked. If $q_Z\geq q_Y\geq q_X+2$, then $q_Z-q_X\geq(q_Z-q_Y)+2 \geq2\lfloor(q_Z-q_Y)/2\rfloor+2$. Under the second hypothesis, if $q_Y>q_Z$ and $g=q_Y-q_Z$, then $q_Y-q_X\geq g+3\geq2\lceil g/2\rceil+2$. If $q_Z\geq q_Y$, then instead $q_Z-q_X\geq(q_Z-q_Y)+3\geq2\lfloor(q_Z-q_Y)/2\rfloor+2$.
\end{proof}

For the remaining comparisons, we keep $\rho$ as a variable when calculating the coefficients of $U^p(X\mid q)-U^p(Y\mid q)$. For each private vector considered below, these calculations give a threshold $\theta(q)\in[0,1]$: ballot $X$ strictly beats $Y$ for sufficiently small positive $p$ when $\rho<\theta(q)$, while $Y$ strictly beats $X$ when $\rho>\theta(q)$. The threshold is read directly from the first coefficient that is not the zero expression. If this coefficient has a zero in $(0,1)$, that zero is the threshold. A threshold of zero means that $Y$ beats $X$ for every $\rho\in(0,1)$; a threshold of one means that $X$ beats $Y$. At $\rho=\theta(q)$, later coefficients may decide the comparison. Thus a strict top response requires $\rho\leq\theta(q)$, and a strict middle response requires $\rho\geq\theta(q)$.

Condition \eqref{eq:app-C4} gives threshold zero. At the first contributing tremble order, ballot $X$ still cannot make $X$ a winner, whereas ballot $Y$ can replace a unique $Z$-winner by a $Y$--$Z$ tie, or replace that tie by a unique $Y$-winner. By \eqref{eq:app-C3}, the leading $X-Y$ coefficient is therefore a positive multiple of $-\rho$, so $\theta(q)=0$.

For threshold one, suppose $q_Y>q_X$ and either $q_X>q_Z$ or $q_X=q_Z$ with $q_Y-q_X$ odd. Put $\eta=\lfloor(q_Y-q_X)/2\rfloor$. After at most $\eta$ trembles,
\[
v_Y-v_Z\geq q_Y-q_Z-2\eta\geq1.
\]
Thus $Z$ cannot win under either contemplated ballot. With fewer than $\eta$ trembles, $Y$ also remains at least two above $X$, so both ballots leave $Y$ as the unique winner. At order $\eta$, every payoff difference is a nonnegative multiple of $1-\rho$, and the flow $\eta(Y\mapsto X)$ gives a strictly positive contribution. The leading coefficient is therefore a positive multiple of $1-\rho$, giving $\theta(q)=1$.

We now calculate the remaining thresholds. All private vectors below use $(X,Y,Z)$-coordinates. One tremble changes the two gaps as follows:
\begin{equation}
\begin{array}{c|rrrrrr}
&Y\mapsto X&Y\mapsto Z&Z\mapsto X&X\mapsto Z&X\mapsto Y&Z\mapsto Y\\ \hline
v_Y-v_X&-2&-1&-1&1&2&1\\
v_Y-v_Z&-1&-2&1&-1&1&2.
\end{array}
\label{eq:app-C5}
\end{equation}
Each tremble reduces either lead of $Y$ by at most two. For each row below, before the stated order $Y$ remains the unique winner under both contemplated ballots, so every earlier coefficient vanishes. At the stated order, the gap changes in \eqref{eq:app-C5} give exactly the contributing flows listed below. Summing their multiplicities and payoff differences using Lemma~\ref{lem:directed-flow-expansion} gives the displayed coefficients. Each coefficient is shown up to a strictly positive factor independent of $\rho$.

Here $h\in\mathbb N_0$. The first three rows have $k\geq2$, and the fourth has $k\geq1$.
\begingroup
\scriptsize
\begin{equation}
\begin{array}{c|c|l}
q&\text{Order}&\text{Coefficient}\\ \hline
(h,h+2k+1,h+2)&k&
(h+k+2)-[2(k+1)h+(k+2)^2]\rho\\
(h,h+2k+2,h+1)&k+1&
2(k+2)h+k^2+6k+7\\
&&-[(4k+7)h+2k^2+10k+11]\rho\\
(h+1,h+2k+1,h+1)&k&
2(k+1)h+(k+1)(k+4)\\
&&-[(2k+3)h+k^2+6k+6]\rho\\
(h+1,h+2k,h+2)&k-1&1-2\rho.
\end{array}
\label{eq:app-C6}
\end{equation}
\endgroup
The contributing flows for these four rows are, respectively:
\begin{enumerate}
\item $k(Y\mapsto Z)$, $(k-1)(Y\mapsto Z)+(Y\mapsto X)$, $k(Y\mapsto X)$, \\
$(k-1)(Y\mapsto Z)+(X\mapsto Z)$.
\item $(k+1)(Y\mapsto Z)$, $k(Y\mapsto Z)+(Y\mapsto X)$, \\
$k(Y\mapsto X)+(Z\mapsto X)$, $k(Y\mapsto X)+(Y\mapsto Z)$, \\
$(k+1)(Y\mapsto X)$, $k(Y\mapsto Z)+(X\mapsto Z)$.
\item $k(Y\mapsto Z)$, $(k-1)(Y\mapsto X)+(Z\mapsto X)$, \\
$(k-1)(Y\mapsto X)+(Y\mapsto Z)$, $k(Y\mapsto X)$.
\item $(k-1)(Y\mapsto Z)$, $(k-1)(Y\mapsto X)$.
\end{enumerate}
A flow requiring a ballot from a candidate with zero support is omitted. In the fourth row with $k=1$, both listed expressions describe the same empty flow, which is counted once; its margin is $(1-2\rho)/3$.

To illustrate the calculation, the four flows in the first row have multiplicities
\[
\begin{aligned}
&\binom{h+2k+1}{k},\qquad
k\binom{h+2k+1}{k},\\
&\binom{h+2k+1}{k},\qquad
h\binom{h+2k+1}{k-1},
\end{aligned}
\]
and respective payoff differences
\[
-\frac{\rho}{2},\qquad
-\frac{\rho}{2},\qquad
\frac{1-\rho}{2},\qquad
-\frac{\rho}{2}.
\]
Multiplying by the common probability factor $2^{-k}$ and adding gives the leading coefficient
\[
\begin{aligned}
\frac{\binom{h+2k+1}{k}}{2^{k+1}(h+k+2)}
\bigl[&(h+k+2)\\
&-[2(k+1)h+(k+2)^2]\rho\bigr],
\end{aligned}
\]
which is the first row of \eqref{eq:app-C6} with its positive factor restored.

For $k=1$, the first three rows must be calculated separately because some terminal vectors give three-way ties. The coefficients, again up to positive factors independent of $\rho$, are
\begingroup
\scriptsize
\begin{equation}
\begin{array}{c|c|l}
q&\text{Order}&\text{Coefficient}\\ \hline
(h,h+3,h+2)&1&
2(h+3)-(10h+21)\rho\\
(h,h+4,h+1)&2&
4(4h+9)-(29h+57)\rho\\
(h+1,h+3,h+1)&1&
(11h+27)-(13h+33)\rho.
\end{array}
\label{eq:app-C7}
\end{equation}
\endgroup
The contributing flows in these three rows are:
\begin{enumerate}
\item $Y\mapsto Z$, $Y\mapsto X$, $X\mapsto Z$.
\item $2(Y\mapsto Z)$, $(Y\mapsto X)+(Z\mapsto X)$, $(Y\mapsto X)+(Y\mapsto Z)$, $2(Y\mapsto X)$, $(X\mapsto Z)+(Y\mapsto Z)$.
\item $Z\mapsto X$, $Y\mapsto Z$, $Y\mapsto X$.
\end{enumerate}

The first three rows of \eqref{eq:app-C6}, using \eqref{eq:app-C7} when $k=1$, therefore give
\begin{equation}
\begin{array}{c|c}
\text{Row}&\text{Threshold}\\ \hline
\text{First}&\theta(q)<1/2\\
\text{Second}&1/2<\theta(q)<2/3\\
\text{Third}&\theta(q)>2/3.
\end{array}
\label{eq:app-C8}
\end{equation}
These bounds can be checked by substituting $\rho=1/2$ and $\rho=2/3$ into the displayed coefficients. For $k\geq2$, the first coefficient at $1/2$ is
\[
-\frac{k(2h+k+2)}2<0.
\]
The second coefficient is positive at $1/2$ and negative at $2/3$, with values
\[
\frac{h+2k+3}{2}>0,
\qquad
-\frac{(k+1)(2h+k+1)}3<0.
\]
The third coefficient at $2/3$ is
\[
\frac{k(2h+k+3)}3>0.
\]
For $k=1$, the corresponding values are
\[
-\frac{6h+9}{2},\qquad
\frac{3h+15}{2},\qquad
-\frac{10h+6}{3},\qquad
\frac{7h+15}{3},
\]
with the same respective signs. Every displayed coefficient decreases as $\rho$ increases, proving \eqref{eq:app-C8}.

\begin{lemma}
\label{lem:three-active-conditional-transfer}
Let $c\in\{Y,Z\}$ and $q_c\geq1$. With $q'=q+\mathbf e_X-\mathbf e_c$, it is impossible that
\[
\operatorname{BR}(q)=X,
\qquad
\operatorname{BR}(q')=Y.
\]
\end{lemma}

\begin{proof}
The middle response and Lemma~\ref{lem:three-active-middle-score} imply $q'_Y>q'_X$. Put $q=(h,h+g,h+\delta)$. If $c=Z$, then $g\geq2$ and
\begin{equation}
\begin{array}{c|c|c}
\text{case}&\theta(q)&\theta(q')\\ \hline
\delta\leq1&\text{not needed}&1\\
\delta\geq3&0&\text{not needed}\\
\delta=2,\ g\text{ even}&0&1\\
\delta=2,\ g=2k+1,\ k\geq1&<1/2&>2/3.
\end{array}
\label{eq:app-C9}
\end{equation}
If $c=Y$, then $g\geq3$ and
\begin{equation}
\begin{array}{c|c|c}
\text{case}&\theta(q)&\theta(q')\\ \hline
\delta\leq0&\text{not needed}&1\\
\delta\geq3&0&\text{not needed}\\
\delta=1,\ g\text{ odd}&\text{not needed}&1\\
\delta=1,\ g=2(k+1),\ k\geq1&<2/3&>2/3\\
\delta=2,\ g\text{ even}&0&\text{not needed}\\
\delta=2,\ g=2k+1,\ k\geq1&<1/2&1/2.
\end{array}
\label{eq:app-C10}
\end{equation}
For a zero entry with $\delta\geq3$, one of the two score conditions in Lemma~\ref{lem:three-active-score-obstructions} holds and hence so does \eqref{eq:app-C4}. For a zero entry with $\delta=2$ and even $g$, condition \eqref{eq:app-C4} holds with $\eta=(g-2)/2$. Every one entry has $q'_X>q'_Z$, or $q'_X=q'_Z$ with an odd $q'_Y-q'_X$, so the preceding comparison applies. The remaining bounds follow from \eqref{eq:app-C6}--\eqref{eq:app-C8}. The cases are disjoint and exhaustive. If $\theta(q)=0$, the top response at $q$ is impossible because $\rho>0$. If $\theta(q')=1$, the middle response at $q'$ is impossible because $\rho<1$. In every remaining row, $\theta(q)<\theta(q')$. A top response at $q$ would require $\rho\leq\theta(q)$, while a middle response at $q'$ would require $\rho\geq\theta(q')$, which is impossible even if a leading coefficient vanishes at one of the thresholds.
\end{proof}

We require one final comparison. The only new leading coefficients are those at, for $h\in\mathbb N_0$,
\[
(h,h+2j+1,h+1),\qquad (h,h+2j+2,h).
\]
Their thresholds are respectively $1/2$ and
\begin{equation}
\frac{11h+16}{13h+20}\quad(j=0),
\qquad
\frac{2(j+2)h+(j+2)(j+3)}{(2j+5)h+(j+2)(j+4)}\quad(j\geq1).
\label{eq:app-C11}
\end{equation}
For the first vector, when $j\geq1$, the complete first contributing order consists of $j(YZ)$ and $j(YX)$; for $j=0$, the order-zero coefficient is $(1-2\rho)/3$. For the second it consists, when $j\geq1$, of $(j+1)(YZ)$, $j(YX)+ZX$, $j(YX)+YZ$, and $(j+1)(YX)$; when $j=0$ it consists of $YX,YZ,ZX$. The flow formula gives \eqref{eq:app-C11}. Subtracting $1/2$ has positive numerator $9h+12$ for $j=0$ and $(2j+3)h+(j+2)^2$ for $j\geq1$.

\begin{lemma}
\label{lem:three-active-middle-transfer}
If $q_Y\geq\max\{q_X,q_Z\}$ and $q_Z\geq1$, it is impossible that
\[
\operatorname{BR}(q)=X,
\qquad
\operatorname{BR}(q+\mathbf e_Y-\mathbf e_Z)=Y.
\]
\end{lemma}

\begin{proof}
Put $q=(h,h+g,h+\delta)$ and $q'=q+\mathbf e_Y-\mathbf e_Z$. The hypotheses give $g\geq\max\{0,\delta\}$ and $h+\delta\geq1$. The preceding calculations give
\begin{equation}
\begin{array}{c|c|c}
\text{case}&\theta(q)&\theta(q')\\ \hline
\delta\leq0&\text{not needed}&1\\
\delta\geq3&0&\text{not needed}\\
\delta=2,\ g\text{ even}&0&1/2\\
\delta=2,\ g\text{ odd}&<1/2&>1/2\\
\delta=1,\ g\text{ even}&\text{not needed}&1\\
\delta=1,\ g\text{ odd}&1/2&>1/2.
\end{array}
\label{eq:app-C12}
\end{equation}
When $\delta=2$ and $g=2k+1$, we have $k\geq1$ and the thresholds are the first and second rows of \eqref{eq:app-C8}. When $\delta=1$ and $g$ is odd, \eqref{eq:app-C11} applies, including its separate $g=1$ value. For $\delta\geq3$, both $Y,Z$ are at least three above $X$, giving threshold zero. For $\delta=2$ and even $g$, condition \eqref{eq:app-C4} gives threshold zero at $q$, and the first family preceding \eqref{eq:app-C11} gives threshold $1/2$ at $q'$. Each threshold-one entry has $q'_X>q'_Z$, or $q'_X=q'_Z$ with an odd $q'_Y-q'_X$.

If $\theta(q)=0$, the top response is impossible because $\rho>0$. If $\theta(q')=1$, the middle response is impossible because $\rho<1$. Otherwise $\theta(q)\leq1/2<\theta(q')$, whereas the prescribed responses require $\rho\leq\theta(q)$ and $\rho\geq\theta(q')$. This is impossible. In particular, the zero-count boundary $q=(0,1,1)$ is included: \eqref{eq:app-C11} gives thresholds $1/2$ and $4/5$.
\end{proof}

\endgroup

\subsection{Excluding the Two Stay Patterns}

Now we turn to the remaining stay patterns and show that also they indeed cannot happen.

\begin{proof}
Suppose a period-three cycle has at most three active voters. Lemma~\ref{lem:condensed-active-path} excludes fewer than three and leaves only the stay patterns $(1,1,1)$ and $(2,1,0)$. These numbers count the active voters who keep their ballots in each of the three transitions.

Let $z=(z_A,z_B,z_C)$ be the fixed counts of the inactive voters. Write $s_i$ for the candidate voter $i$ supports twice and $o_i$ for the candidate she supports once. These are her favorite and middle candidate, in some order. Each table entry lists the other two active ballots and the required next ballot. For example, $\{A,B\}\mapsto C$ means that the voter must respond with $C$ to the private vector $z+\mathbf e_A+\mathbf e_B$. All braces denote multisets, and column $t$ gives the response from time $t$ to time $t+1$, with time read modulo three.

Consider first the stay pattern $(1,1,1)$. Label the voters so that voters $1,2,3$ stay in transitions $0,1,2$, respectively. Their ballot paths are
\[
(s_1,s_1,o_1),\qquad
(o_2,s_2,s_2),\qquad
(s_3,o_3,s_3).
\]
The required responses are
\begin{equation}
\begin{array}{c|ccc}
&t=0&t=1&t=2\\ \hline
1&\{o_2,s_3\}\mapsto s_1
 &\{s_2,o_3\}\mapsto o_1
 &\{s_2,s_3\}\mapsto s_1\\
2&\{s_1,s_3\}\mapsto s_2
 &\{s_1,o_3\}\mapsto s_2
 &\{o_1,s_3\}\mapsto o_2\\
3&\{s_1,o_2\}\mapsto o_3
 &\{s_1,s_2\}\mapsto s_3
 &\{o_1,s_2\}\mapsto s_3.
\end{array}
\label{eq:app-C13}
\end{equation}
The pairs $(s_i,o_i)$ must be distinct. If $(s_2,o_2)=(s_3,o_3)$, the first two entries of voter $1$'s row have identical other ballots but different required responses. Equality of the pairs of voters $1,3$ or $1,2$ gives the same contradiction for voter $2$ or $3$.

The exhaustive reduction in Subsection~\ref{app:three-active-t1} applies Lemma~\ref{lem:three-active-conditional-transfer} to these responses. After relabeling candidates, it leaves only
\[
\begin{array}{c|c}
((s_1,o_1),(s_2,o_2),(s_3,o_3))
&\text{repeated candidates}\\ \hline
((A,B),(B,C),(C,A))
&\text{all middle-ranked}\\
((A,B),(C,A),(B,C))
&\text{all favorite}.
\end{array}
\]
In the first case, the middle responses of voters $1,2,3$ face $\{C,C\}$, $\{A,A\}$, and $\{B,B\}$, respectively. Lemma~\ref{lem:three-active-middle-score} therefore requires
\[
z_A>z_B,\qquad z_B>z_C,\qquad z_C>z_A,
\]
which is impossible. In the second case, their middle responses face $\{C,C\}$, $\{B,B\}$, and $\{A,A\}$, giving
\[
z_B>z_A,\qquad z_A>z_C,\qquad z_C>z_B.
\]
This is also impossible, excluding $(1,1,1)$.

It remains to consider $(2,1,0)$. Label the voters so that voters $1,2$ stay in transition $0$ and voter $3$ stays in transition $1$. Their paths are
\[
(s_1,s_1,o_1),\qquad
(s_2,s_2,o_2),\qquad
(o_3,s_3,s_3),
\]
and their required responses are
\begin{equation}
\begin{array}{c|ccc}
&t=0&t=1&t=2\\ \hline
1&\{s_2,o_3\}\mapsto s_1
 &\{s_2,s_3\}\mapsto o_1
 &\{o_2,s_3\}\mapsto s_1\\
2&\{s_1,o_3\}\mapsto s_2
 &\{s_1,s_3\}\mapsto o_2
 &\{o_1,s_3\}\mapsto s_2\\
3&\{s_1,s_2\}\mapsto s_3
 &\{s_1,s_2\}\mapsto s_3
 &\{o_1,o_2\}\mapsto o_3.
\end{array}
\label{eq:app-C14}
\end{equation}
Suppose $s_1$ is the favorite. If $s_1=s_3$, changing the other ballot $o_3$ to $s_3$ takes voter $1$ from a favorite response in column $0$ to a middle response in column $1$. If $s_1=s_2$, changing $o_2$ to $s_2$ gives the same forbidden comparison from column $2$ to column $1$. Lemma~\ref{lem:three-active-conditional-transfer} excludes both possibilities. If $s_1$ is the middle candidate, the reverse comparisons give the corresponding restrictions on $o_1$. Thus
\[
\begin{aligned}
s_1\text{ favorite}
&\Longrightarrow s_1\notin\{s_2,s_3\},\\
s_1\text{ middle}
&\Longrightarrow o_1\notin\{o_2,o_3\}.
\end{aligned}
\]
The same restrictions hold with voters $1,2$ exchanged.

Voter $3$ requires different responses to $\{s_1,s_2\}$ and $\{o_1,o_2\}$, so these multisets must differ. If $s_3$ is her favorite, Lemma~\ref{lem:three-active-conditional-transfer} also forbids obtaining $\{o_1,o_2\}$ from $\{s_1,s_2\}$ by changing one ballot to $s_3$. If $s_3$ is her middle candidate, it forbids the reverse change by moving one ballot to $o_3$. In each comparison, the replaced ballot is present, so the lemma's positive-source-count condition is satisfied.

Finally, every prescribed middle response requires $q_Y\geq q_X+1$ by Lemma~\ref{lem:three-active-middle-score}. Applying this inequality to the private counts in \eqref{eq:app-C14} gives restrictions on $z$. Together, these conditions leave the following six cases, up to candidate relabeling and exchange of voters $1,2$. Subsection~\ref{app:three-active-enumeration} proves this reduction exhaustive. Candidate pairs and rankings are listed in voter order; for example, $BAC$ means $B\succ A\succ C$.

\begingroup
\small
\begin{equation}
\begin{array}{c|c|c|l}
\text{row}&\text{candidate pairs}&\text{rankings}
&\text{required inequalities}\\ \hline
1&((A,B),(B,C),(C,B))&(ABC,BCA,BCA)
&\begin{aligned}[t]
z_C&\geq z_B+2,\\
z_B&\geq z_A
\end{aligned}\\[3pt]
2&((A,B),(C,B),(B,A))&(ABC,CBA,ABC)
&\begin{aligned}[t]
z_B&\geq z_A+2,\\
z_B&\geq z_C
\end{aligned}\\[3pt]
3&((A,B),(B,C),(B,A))&(BAC,CBA,BAC)
&\begin{aligned}[t]
z_A&\geq z_B+2,\\
z_B&\geq z_C+1
\end{aligned}\\[3pt]
4&((A,B),(B,C),(A,C))&(BAC,BCA,CAB)
&\begin{aligned}[t]
z_A&\geq z_C\geq z_B+1,\\
z_A&\geq z_B+2
\end{aligned}\\[3pt]
5&((A,B),(B,C),(A,C))&(BAC,BCA,ACB)
&z_C\geq z_A\geq z_B+2\\[3pt]
6&((A,B),(A,C),(B,A))&(BAC,CAB,BAC)
&\begin{aligned}[t]
z_A&\geq z_B+2,\\
z_A&\geq z_C+1.
\end{aligned}
\end{array}
\label{eq:app-C15}
\end{equation}
\endgroup

Four cases are excluded directly by Lemma~\ref{lem:three-active-score-obstructions}. The following entries of \eqref{eq:app-C14} require a favorite response $X$. Private vectors are written in the focal voter's $(X,Y,Z)$ order.
\[
\begin{array}{c|c|c|c}
\text{Row}&\text{Voter, Time}&(X,Y,Z)&q\\ \hline
1&(1,0)&(A,B,C)&(z_A,z_B+2,z_C)\\
3&(2,1)&(C,B,A)&(z_C,z_B+1,z_A+1)\\
4&(2,0)&(B,C,A)&(z_B,z_C+1,z_A+1)\\
5&(2,0)&(B,C,A)&(z_B,z_C+1,z_A+1).
\end{array}
\]
The inequalities in \eqref{eq:app-C15} give $q_Z\geq q_Y\geq q_X+2$ in rows $1,3,4$ and $\min\{q_Y,q_Z\}\geq q_X+3$ in row $5$. In each case, the required favorite response is therefore impossible.

In row $2$, voter $2$ ranks $C\succ B\succ A$. Columns $0$ and $1$ of \eqref{eq:app-C14} require
\[
\{A,A\}\mapsto C,
\qquad
\{A,B\}\mapsto B.
\]
The first private vector, in $(C,B,A)$ order, is
\[
q=(z_C,z_B,z_A+2).
\]
The inequalities in \eqref{eq:app-C15} give $q_Y\geq\max\{q_X,q_Z\}$ and $q_Z=z_A+2>0$. Changing one other ballot from $A$ to $B$ increases the middle score and decreases the bottom score. The required change from the favorite response $C$ to the middle response $B$ contradicts Lemma~\ref{lem:three-active-middle-transfer}.

In row $6$, voter $2$ ranks $C\succ A\succ B$. Columns $1$ and $0$ require
\[
\{A,B\}\mapsto C,
\qquad
\{A,A\}\mapsto A.
\]
The first private vector, in $(C,A,B)$ order, is
\[
q=(z_C,z_A+1,z_B+1).
\]
Again, \eqref{eq:app-C15} gives $q_Y\geq\max\{q_X,q_Z\}$ and $q_Z=z_B+1>0$. Changing one other ballot from $B$ to $A$ contradicts the same lemma.

Thus all six cases are impossible, excluding $(2,1,0)$. Neither three-active stay pattern can occur, and Lemma~\ref{lem:condensed-active-path} already excludes fewer active voters. Every period-three cycle therefore has at least four active voters. We remark that no restriction was placed on the inactive counts $z$.
\end{proof}

We now verify that the proof covers all possible choice of candidate pairs and rankings: first for the stay pattern $(1,1,1)$, then for $(2,1,0)$.

\begingroup

\subsection[The (1,1,1) Reduction]{The $(1,1,1)$ Reduction}
\label{app:three-active-t1}

Use the paths and required responses in \eqref{eq:app-C13}. For size-two multisets, $P\xrightarrow{c}Q$ means that $Q$ is obtained from $P$ by replacing one ballot not equal to $c$ by a ballot for $c$. This arrow describes a change of another voter's ballot, not the focal voter's response. The comparison need not follow the chronological order of the trajectory. If $s_i$ is the favorite, Lemma~\ref{lem:three-active-conditional-transfer} forbids a change toward $s_i$ from a multiset prescribing $s_i$ to one prescribing $o_i$. If $s_i$ is the middle candidate, it forbids a change toward $o_i$ from a multiset prescribing $o_i$ to one prescribing $s_i$. The replaced ballot is present in the source multiset, as the lemma requires.

Relabel candidates so that $(s_1,o_1)=(A,B)$. Since the three ordered pairs are distinct, there are five choices for $(s_2,o_2)$ and then four for $(s_3,o_3)$. Cyclically rotating voters together with time, and then relabeling candidates to restore $(s_1,o_1)=(A,B)$, groups these $20$ choices as follows. Every entry lists the ordered pairs $((s_2,o_2),(s_3,o_3))$; braces will be reserved for multisets of ballots.
\begingroup
\small
\begin{equation}
\begin{array}{c|c|l}
\text{Class}&\text{Representative}&\text{Other members}\\ \hline
1&((A,C),(B,A))&((B,A),(B,C)),\ ((C,A),(A,C))\\
2&((A,C),(B,C))&((C,A),(C,B)),\ ((C,B),(A,C))\\
3&((A,C),(C,A))&((B,A),(A,C)),\ ((B,C),(B,A))\\
4&((A,C),(C,B))&((B,C),(A,C)),\ ((C,B),(C,A))\\
5&((B,A),(C,A))&((B,C),(C,B)),\ ((C,B),(B,A))\\
6&((B,A),(C,B))&((C,A),(B,A)),\ ((C,B),(B,C))\\
7&((B,C),(C,A))&-\\
8&((C,A),(B,C))&-
\end{array}
\label{eq:app-C16}
\end{equation}
\endgroup

The first six classes are excluded by the following comparisons in their representatives. In each row, the first change is forbidden if $s_i$ is the favorite and the second if $s_i$ is the middle candidate. The required responses are those in \eqref{eq:app-C13}.
\begingroup
\small
\[
\begin{array}{c|c|l}
\text{Class}&i&\text{Forbidden changes}\\ \hline
1&1&\{A,B\}\xrightarrow{A}\{A,A\},\quad
     \{A,A\}\xrightarrow{B}\{A,B\}\\
2&1&\{B,C\}\xrightarrow{A}\{A,C\},\quad
     \{A,C\}\xrightarrow{B}\{B,C\}\\
3&3&\{A,A\}\xrightarrow{C}\{A,C\},\quad
     \{A,C\}\xrightarrow{A}\{A,A\}\\
4&3&\{A,A\}\xrightarrow{C}\{A,C\},\quad
     \{A,C\}\xrightarrow{B}\{A,B\}\\
5&2&\{A,C\}\xrightarrow{B}\{B,C\},\quad
     \{B,C\}\xrightarrow{A}\{A,C\}\\
6&2&\{A,C\}\xrightarrow{B}\{B,C\},\quad
     \{B,C\}\xrightarrow{A}\{A,C\}
\end{array}
\]
\endgroup

For example, in class $1$, voter $1$ must respond with $A$ to $\{A,B\}$ and with $B$ to $\{A,A\}$. If $A$ is her favorite, replacing the other ballot $B$ by $A$ gives a forbidden favorite-to-middle response change. If $B$ is her favorite, the reverse comparison gives the same contradiction. Thus neither ranking is possible. The remaining five rows work in the same way.

In class $7$, the pairs are $((A,B),(B,C),(C,A))$. If any repeated candidate is the corresponding voter's favorite, the forbidden changes for voters $1,2,3$ are, respectively, $\{B,C\}\xrightarrow{A}\{A,B\}$, $\{A,C\}\xrightarrow{B}\{B,C\}$, and $\{A,B\}\xrightarrow{C}\{A,C\}$. Hence every $s_i$ must be the middle candidate.

In class $8$, the pairs are $((A,B),(C,A),(B,C))$. If any repeated candidate is the corresponding voter's middle candidate, the forbidden changes for voters $1,2,3$ are, respectively, $\{C,C\}\xrightarrow{B}\{B,C\}$, $\{B,B\}\xrightarrow{A}\{A,B\}$, and $\{A,A\}\xrightarrow{C}\{A,C\}$. Hence every $s_i$ must be the favorite.

The first six classes contain three choices each, and the last two contain one each. All $20$ choices therefore appear in \eqref{eq:app-C16}, and the only possibilities not excluded by Lemma~\ref{lem:three-active-conditional-transfer} are the two used in the proof above.

\subsection{Exhaustiveness of Table~\ref{eq:app-C15}}
\label{app:three-active-enumeration}

Use the paths and restrictions from \eqref{eq:app-C14} and retain the multiset-change notation above. If exactly one of $s_1,s_2$ is the favorite, exchange voters $1,2$ so that $s_1$ is the middle candidate and $s_2$ is the favorite. Then relabel candidates so that $(s_1,o_1)=(A,B)$. Every prescribed middle response gives an integer inequality through Lemma~\ref{lem:three-active-middle-score}.

Every choice not excluded below appears in the following ten rows. The second column lists the ordered pairs $((s_2,o_2),(s_3,o_3))$. The last three columns indicate whether each repeated candidate is its voter's favorite or middle candidate.
\begingroup
\small
\begin{equation}
\begin{array}{c|c|ccc}
\text{Row}&((s_2,o_2),(s_3,o_3))&s_1&s_2&s_3\\ \hline
1&((B,C),(C,B))&\text{Favorite}&\text{Favorite}&\text{Middle}\\
2&((C,A),(B,A))&\text{Favorite}&\text{Favorite}&\text{Middle}\\
3&((C,B),(B,A))&\text{Favorite}&\text{Favorite}&\text{Middle}\\
4&((C,B),(B,C))&\text{Favorite}&\text{Favorite}&\text{Middle}\\
5&((A,C),(B,A))&\text{Middle}&\text{Middle}&\text{Favorite}\\
6&((A,C),(C,A))&\text{Middle}&\text{Middle}&\text{Favorite}\\
7&((B,C),(B,A))&\text{Middle}&\text{Middle}&\text{Favorite}\\
8&((C,A),(A,C))&\text{Middle}&\text{Middle}&\text{Favorite}\\
9&((B,C),(A,C))&\text{Middle}&\text{Favorite}&\text{Favorite}\\
10&((B,C),(A,C))&\text{Middle}&\text{Favorite}&\text{Middle}
\end{array}
\label{eq:app-C17}
\end{equation}
\endgroup

We check that no choice is missing. Suppose first that $s_1,s_2$ are both favorites. The restrictions on voters $1,2$ force $\{s_2,s_3\}=\{B,C\}$. If $(s_2,s_3)=(B,C)$, then $o_2=C$, since otherwise voter $3$ would face identical private counts in columns $0$ and $2$ of \eqref{eq:app-C14}, with different required responses. If $s_3$ is the favorite, voter $3$ is excluded by $\{A,B\}\xrightarrow{C}\{B,C\}$. If $s_3$ is the middle candidate, the choice $o_3=A$ is excluded by $\{B,C\}\xrightarrow{A}\{A,B\}$. Thus $o_3=B$, giving row $1$.

If $(s_2,s_3)=(C,B)$, then $o_2\in\{A,B\}$ and $o_3\in\{A,C\}$. If $s_3$ is the favorite, the four choices $(o_2,o_3)=(A,A),(B,A),(A,C),(B,C)$ give, respectively,
\[
\begin{array}{c|l}
(o_2,o_3)&\text{Forced inequalities}\\ \hline
(A,A)&z_A\geq z_B+1\\
(B,A)&z_A\geq z_B+3\\
(A,C)&z_A\geq z_C\geq z_B+2\\
(B,C)&z_B\geq z_C\geq z_B+3
\end{array}
\]
The first three contradict $z_B\geq z_A$, which follows from voter $1$'s middle response. The fourth itself is clearly impossible. If $s_3$ is the middle candidate, the choice $(o_2,o_3)=(A,C)$ is excluded by $\{A,B\}\xrightarrow{C}\{A,C\}$ for voter $3$, leaving rows $2$--$4$.

Suppose next that $s_1,s_2$ are both middle candidates. Their restrictions force $\{o_2,o_3\}=\{A,C\}$. If $(o_2,o_3)=(C,A)$ and $s_3$ is the middle candidate, the four choices for $(s_2,s_3)$ give
\[
\begin{array}{c|l}
(s_2,s_3)&\text{Forced inequalities}\\ \hline
(A,B)&z_B\geq z_A+3,\quad z_A\geq z_B+2\\
(A,C)&z_C\geq z_A+3,\quad z_A\geq z_C+2\\
(B,B)&z_B\geq z_A+1,\quad z_A\geq z_B+2\\
(B,C)&z_C\geq z_A+2,\quad z_A\geq z_B+1,\quad z_B\geq z_C+1
\end{array}
\]
All four are impossible. If $s_3$ is the favorite, the choice $(s_2,s_3)=(B,C)$ is excluded by $\{A,B\}\xrightarrow{C}\{B,C\}$ for voter $3$. The other choices give rows $5$--$7$.

If $(o_2,o_3)=(A,C)$, then $s_2=C$, since choosing $s_2=B$ would make voter $3$ face identical private counts in columns $0$ and $2$ of \eqref{eq:app-C14}, with different required responses. If $s_3$ is the middle candidate, the two choices $s_3=A$ and $s_3=B$ give, respectively,
\[
\begin{array}{c|l}
s_3&\text{Forced inequalities}\\ \hline
A&z_A\geq z_C+1,\quad z_C\geq z_A+1\\
B&z_A\geq z_B+1,\quad z_B\geq z_C+2,\quad z_C\geq z_A+1
\end{array}
\]
Both are impossible. If $s_3$ is the favorite, the choice $s_3=B$ is excluded by $\{A,C\}\xrightarrow{B}\{A,B\}$ for voter $3$. Only $s_3=A$ remains, giving row $8$.

Finally, suppose that $s_1$ is the middle candidate and $s_2$ is the favorite. The restrictions give $s_2\in\{B,C\}$. If $s_2=B$, then $o_2=C$, since choosing $o_2=B$ would make voter $3$ face identical private counts in columns $0$ and $2$ of \eqref{eq:app-C14}, with different required responses, and the remaining restrictions give $\{s_3,o_3\}=\{A,C\}$. The choice $(s_3,o_3)=(C,A)$ is excluded by $\{A,B\}\xrightarrow{C}\{B,C\}$ when $s_3$ is the favorite and by $\{B,C\}\xrightarrow{A}\{A,B\}$ when it is the middle candidate. The other choice, $(s_3,o_3)=(A,C)$, gives rows $9$ and $10$.

If $s_2=C$, the restrictions force $o_2=A$ and leave $(s_3,o_3)\in\{(A,C),(B,A),(B,C)\}$. Each possibility is excluded as follows. Every displayed ballot change concerns voter $3$.
\begingroup
\small
\[
\begin{array}{c|c|l}
(s_3,o_3)&s_3&\text{Contradiction}\\ \hline
(A,C)&\text{Favorite}&z_A+2\leq z_C\leq z_A+1\\
(A,C)&\text{Middle}&\{A,B\}\xrightarrow{C}\{A,C\}\\
(B,A)&\text{Favorite}&\{A,C\}\xrightarrow{B}\{A,B\}\\
(B,A)&\text{Middle}&z_A\geq z_B+1,\quad z_B\geq z_A+2\\
(B,C)&\text{Favorite}&\{A,C\}\xrightarrow{B}\{A,B\}\\
(B,C)&\text{Middle}&\{A,B\}\xrightarrow{C}\{A,C\}
\end{array}
\]
\endgroup
This exhausts all choices and proves that every remaining case appears in \eqref{eq:app-C17}.

The ten rows reduce to the six rows of \eqref{eq:app-C15} by the following relabelings. The first and last columns give the row numbers in \eqref{eq:app-C17} and \eqref{eq:app-C15}, respectively. Perform any indicated exchange of voters before renaming candidates. Candidate renaming applies simultaneously to the candidate pairs, rankings, and inactive counts.
\begingroup
\small
\begin{equation}
\begin{array}{c|c|c|c}
\text{Row}&\text{Swap }1,2&\text{Candidate renaming}&\text{Row}\\ \hline
1&\text{No}&\text{None}&1\\
2&\text{Yes}&A\mapsto B,\ B\mapsto C,\ C\mapsto A&1\\
3&\text{No}&\text{None}&2\\
4&\text{Yes}&A\leftrightarrow C&2\\
5&\text{No}&\text{None}&6\\
6&\text{Yes}&B\leftrightarrow C&6\\
7&\text{No}&\text{None}&3\\
8&\text{Yes}&A\mapsto B,\ B\mapsto C,\ C\mapsto A&3\\
9&\text{No}&\text{None}&5\\
10&\text{No}&\text{None}&4
\end{array}
\label{eq:app-C18}
\end{equation}
\endgroup

Applying Lemma~\ref{lem:three-active-middle-score} to the required middle responses gives the forced inactive gaps in \eqref{eq:app-C15}. Thus every possible choice of paths and rankings is either excluded above or reduces to one of its six rows, whose contradictions were proved in the preceding subsection.

\endgroup

\section{Proof of Theorem~\ref{ex:six-active-cycle}}
\label{app:six-active-details}
We check that every ballot prescribed by Theorem~\ref{ex:six-active-cycle} is a strict best response at its source profile. This includes voters who keep their ballot in a round, as well as those who switch it.

At sources $S^0,S^1,S^2$, the polls are $(2,2,2)$, $(1,1,4)$, and $(0,0,6)$, respectively. For one representative of each symmetry class, the table compares the prescribed ballot with the other of her favorite and middle candidates. Private vectors use $(A,B,C)$ coordinates. The side voters have utilities $(1,\rho,0)$ in ranking order, where $9/11<\rho<1$; the two $C$-types have utilities $(1,\gamma,0)$ in ranking order, where $1/2<\gamma<1$. Every displayed coefficient is strictly positive throughout these intervals.
\begingroup
\small
\[
\begin{array}{c|c|c|c|c|c}
\text{Source}&\text{Type}&q&\text{Margin}&\text{Order}&
\text{Coefficient}\\ \hline
S^0&A\succ C\succ B&(1,2,2)&C-A&0&(2\rho-1)/3\\
S^0&C\succ A\succ B&(2,2,1)&A-C&0&(2\gamma-1)/3\\
S^1&A\succ C\succ B&(1,1,3)&C-A&1&(11\rho-9)/4\\
S^1&C\succ A\succ B&(0,1,4)&C-A&2&(19-12\gamma)/4\\
S^2&A\succ C\succ B&(0,0,5)&A-C&2&5(1-\rho)/4\\
S^2&C\succ A\succ B&(0,0,5)&C-A&2&5(1-\gamma)/4
\end{array}
\]
\endgroup

The $B$-side comparisons follow by exchanging $A$ and $B$. Lemma~\ref{lem:top-bottom-dominance} makes the favorite ballot strictly better than the least-preferred ballot. A prescribed favorite therefore beats both alternatives, and a prescribed middle ballot that beats the favorite also beats the least-preferred ballot. The table establishes every voter's unique response at all three profiles.

\section{Proof of Example~\ref{ex:truthful-four-cycle}}
\label{app:four-cycle-margins}

We verify Example~\ref{ex:truthful-four-cycle} with the normalized utilities and ballot paths in Figure~\ref{fig:truthful-four-cycle}. The complete profiles, in voter order, are
\[
\begin{aligned}
S^0&=(A,A,B,B,C,C),\\
S^1&=(B,B,B,A,B,C),\\
S^2&=(A,A,B,B,C,A),\\
S^3&=(A,A,B,A,B,C),\qquad S^4=S^0.
\end{aligned}
\]
At each source profile, voter $i$ removes her current ballot and uses $q_i^t=n^t-\mathbf e_{s_i^t}$. Each private vector has total mass five. Table~\ref{tab:four-cycle-margins} compares her prescribed next ballot with the other of her favorite and middle candidates. Voters $1$ and $2$ have identical utilities and ballots, so their comparisons coincide.

Apply Lemma~\ref{lem:directed-flow-expansion} to each displayed private vector and margin. An entry $(d,a)$ means that $H_\eta(c,c';q_i^t)=0$ for every $\eta<d$ and $2^{-d}H_d(c,c';q_i^t)=a>0$. Thus the margin has first nonzero term $ap^d$. Each sum includes all feasible directed flows at its order, with their multinomial multiplicities and uniform averages over winner sets, including mixed trembles and ties.

\begin{table}[!ht]
\centering
\caption{All favorite--middle comparisons for the truthful period-four orbit. Private counts are in $(A,B,C)$ order. The last column gives the first nonzero order and coefficient $(d,a)$ of the chosen-minus-alternative margin, so the margin is $ap^d+\mathcal O(p^{d+1})$. Every preceding coefficient is zero.}
\label{tab:four-cycle-margins}
\begin{tabular}{ccccc}
\toprule
Round & Voter & Private counts & Margin & $(d,a)$\\
\midrule
0 & 1,2 & $(1,2,2)$ & $B-A$ & $(0,1/9)$\\
0 & 3 & $(2,1,2)$ & $B-A$ & $(0,1/6)$\\
0 & 4 & $(2,1,2)$ & $A-B$ & $(0,2/9)$\\
0 & 5 & $(2,2,1)$ & $B-C$ & $(0,1/9)$\\
0 & 6 & $(2,2,1)$ & $C-A$ & $(0,1/6)$\\
\midrule
1 & 1,2 & $(1,3,1)$ & $A-B$ & $(1,5/12)$\\
1 & 3 & $(1,3,1)$ & $B-A$ & $(1,35/16)$\\
1 & 4 & $(0,4,1)$ & $B-A$ & $(2,9/4)$\\
1 & 5 & $(1,3,1)$ & $C-B$ & $(1,5/12)$\\
1 & 6 & $(1,4,0)$ & $A-C$ & $(1,1/4)$\\
\midrule
2 & 1,2 & $(2,2,1)$ & $A-B$ & $(0,1/3)$\\
2 & 3 & $(3,1,1)$ & $B-A$ & $(1,25/16)$\\
2 & 4 & $(3,1,1)$ & $A-B$ & $(1,1/24)$\\
2 & 5 & $(3,2,0)$ & $B-C$ & $(0,1/3)$\\
2 & 6 & $(2,2,1)$ & $C-A$ & $(0,1/6)$\\
\midrule
3 & 1,2 & $(2,2,1)$ & $A-B$ & $(0,1/3)$\\
3 & 3 & $(3,1,1)$ & $B-A$ & $(1,25/16)$\\
3 & 4 & $(2,2,1)$ & $B-A$ & $(0,1/6)$\\
3 & 5 & $(3,1,1)$ & $C-B$ & $(1,7/12)$\\
3 & 6 & $(3,2,0)$ & $C-A$ & $(1,1/16)$\\
\bottomrule
\end{tabular}
\end{table}

Every displayed coefficient is positive. Lemma~\ref{lem:top-bottom-dominance} makes the favorite ballot strictly better than the least-preferred ballot. Hence a prescribed favorite beats both alternatives; a prescribed middle ballot that beats the favorite also beats the least-preferred ballot. The table therefore verifies every voter's unique limiting response at all four source profiles. The initial ballots are the respective favorites, the four complete profiles are distinct, and the initial count vector $(2,2,2)$ occurs only once in the four-round block. Both the complete profile and the sequence of polls consequently have period four. Exactly five voters are active, and the winner sets are $\{A,B,C\},\{B\},\{A\},\{A\}$.

The comparisons do not require the exact utility values in the figure. With each voter's favorite and least-preferred utilities normalized to $1$ and $0$, let $\rho_i$ be her middle utility. The same enumeration verifies the same strict responses throughout the nonempty open region
\[
\begin{gathered}
\frac12<\rho_i<\frac9{11}\quad(i=1,2,5),
\qquad 0<\rho_3<\frac12,\\
\frac9{11}<\rho_4<1,
\qquad 0<\rho_6<\frac27.
\end{gathered}
\]
Indeed, each first coefficient is affine in the relevant middle utility, nonnegative at both stated interval bounds, and positive at at least one of them; every preceding coefficient vanishes throughout the interval. It is therefore positive throughout the open interval. The specified normalized utilities lie strictly inside these intervals. For each utility profile in this region, Lemma~\ref{lem:small-p} also gives the same cycle under $F_p$ for all sufficiently small positive $p$.

\section{Proof of Lemma~\ref{lem:paired-response}}\label{app:paired-response}

\begingroup

This Appendix gives the complete flow calculations for Lemma~\ref{lem:paired-response}. Only signs of leading coefficients matter, so common strictly positive factors are suppressed throughout. All coordinates are in $(A,B,C)$ order.

\begin{proof}[Proof of Lemma~\ref{lem:paired-response}]
Let $h=2s+1\geq7$. A split pair casts $(A,B)$ and a pooled pair casts $(C,C)$. Let $x$ pairs be split and put $d=h-x$. Including the $h$ added $C$-voters, the poll is
\[
(x,x,2(h-x)+h)=(x,x,x+3d).
\]
By symmetry, it is enough to consider a voter with ranking $A\succ C\succ B$ and normalized utilities
\[
 (u(A),u(C),u(B))=(1,\rho,0),\qquad \frac12<\rho<1.
\]
Her private vector is
\begin{equation}
\begin{aligned}
q^{A}&=(x-1,x,x+3d)&&\text{if split, so }x>0,\\
q^{C}&=(x,x,x+3d-1)&&\text{if pooled, so }x<h.
\end{aligned}
\label{eq:app-E1}
\end{equation}
Voting for $A$ strictly dominates voting for $B$ by Lemma~\ref{lem:top-bottom-dominance}, so only $A$ and $C$ need to be compared.

At terminal private counts $v$, both contemplated ballots leave $C$ the unique winner whenever
\[
 v_A\leq v_C-2,
 \qquad
 v_B\leq v_C-1.
\]
A nonzero contribution therefore requires
\begin{equation}
 v_A\geq v_C-1
 \qquad\text{or}\qquad
 v_B\geq v_C.
\label{eq:app-E2}
\end{equation}
The decrease in each lead of $C$ caused by one directed tremble is
\begin{equation}
\begin{array}{c|rrrrrr}
&C\mapsto A&C\mapsto B&B\mapsto A&A\mapsto B&A\mapsto C&B\mapsto C\\ \hline
v_C-v_A&2&1&1&-1&-2&-1\\
v_C-v_B&1&2&-1&1&-1&-2
\end{array}
\label{eq:app-E3}
\end{equation}
No tremble reduces either lead by more than two. A required reduction of $2r$ therefore takes at least $r$ trembles, and at that order every tremble must reduce the lead by two. For a required reduction of $2r-1$, either all $r$ trembles reduce the lead by two or exactly one reduces it by one. Together with~\eqref{eq:app-E3}, this gives the complete first contributing flow lists below and excludes every earlier order pointwise.

If $x$ is odd, $d$ is even. When $d=0$, every pair is split and the no-tremble $A-C$ margin at $q^A=(x-1,x,x)$ is $(1-2\rho)/3<0$, so all pairs pool.

For even $d\geq2$, put $t=3d/2$. At $q^C$, the first contributing order is $t-1$, with the single flow $(t-1)(C\mapsto A)$ and payoff difference $(1-\rho)/2>0$. A pooled pair therefore splits. At $q^A$, the first contributing order is $t$, with flows
\[
 t(C\mapsto A),
 \qquad
 t(C\mapsto B)
\]
Their multiplicities are equal and their payoff differences are $(1-\rho)/2$ and $-\rho/2$, respectively. The sum has sign $1-2\rho<0$, so a split pair pools. Thus every pair changes configuration when $x$ is odd.

Now let $x=2k$, so $d$ is odd. For $d\geq3$, put
\[
 t:=\frac{3d-1}{2}.
\]
At $q^C=(x,x,x+2t)$, the first contributing order is $t$. Dividing multiplicities by $\binom{x+2t}{t}$ gives
\begin{equation}
\begin{array}{c|c|c}
\text{Flow}&\text{Relative multiplicity}&
 \text{Payoff difference}\\ \hline
t(C\mapsto A)&1&1-\rho\\
(t-1)(C\mapsto A)+(C\mapsto B)&t&(1-\rho)/2\\
t(C\mapsto B)&1&-\rho/2\\
(t-1)(C\mapsto A)+(B\mapsto A)&\frac{tx}{x+t+1}&(1-\rho)/2
\end{array}
\label{eq:app-E4}
\end{equation}
The sum is a positive multiple of $\kappa_{h-k}-\rho$, where
\begin{equation}
 \kappa_{h-k}:=
 \frac{2+t+\dfrac{tx}{x+t+1}}
      {3+t+\dfrac{tx}{x+t+1}}.
\label{eq:app-E5}
\end{equation}
Thus a pooled pair splits exactly when $\rho<\kappa_{h-k}$, provided equality is excluded. A row of zero multiplicity is absent.

If $x>0$, the split voter's private counts are $q^A=(x-1,x,x+2t+1)$ and the first contributing order is $t+1$. Dividing multiplicities by $\binom{x+2t+1}{t+1}$ gives
\begin{equation}
\begin{array}{c|c|c}
\text{Flow}&\text{Relative multiplicity}&
 \text{Payoff difference}\\ \hline
(t+1)(C\mapsto A)&1&1-\rho\\
t(C\mapsto A)+(C\mapsto B)&t+1&(1-\rho)/2\\
t(C\mapsto A)+(B\mapsto A)&\frac{x(t+1)}{x+t+1}&(1-\rho)/2\\
(t+1)(C\mapsto B)&1&-\rho/2\\
t(C\mapsto B)+(C\mapsto A)&t+1&-\rho/2\\
t(C\mapsto B)+(A\mapsto B)&\frac{(x-1)(t+1)}{x+t+1}&-\rho/2
\end{array}
\label{eq:app-E6}
\end{equation}
This sum is a positive multiple of $\kappa_k-\rho$, where
\begin{equation}
 \kappa_k:=
 \frac{t+3+\dfrac{x(t+1)}{x+t+1}}
      {2t+5+\dfrac{(2x-1)(t+1)}{x+t+1}}.
\label{eq:app-E7}
\end{equation}
Thus a split pair remains split exactly when $\rho<\kappa_k$, provided equality is excluded. When $x=0$, no split-pair comparison exists; $\kappa_0=1/2$ is only the lower utility bound.

When $d=1$, some generic flow expressions coincide and three-way ties change the payoffs. Each distinct flow must be counted once. The complete calculations at $q^C=(x,x,x+2)$ and $q^A=(x-1,x,x+3)$ are
\begingroup\small
\begin{equation}
\begin{array}{c|c|c}
\multicolumn{3}{c}{q^{C}:\ \text{order }1}\\
\text{Flow}&\text{Multiplicity}&\text{Payoff difference}\\ \hline
C\mapsto A&x+2&1-\rho\\
C\mapsto B&x+2&(1-2\rho)/3\\
B\mapsto A&x&(1-\rho)/2
\end{array}
\label{eq:app-E8}
\end{equation}
\begin{equation}
\begin{array}{c|c|c}
\multicolumn{3}{c}{q^{A}:\ \text{order }2}\\
\text{Flow}&\text{Multiplicity}&\text{Payoff difference}\\ \hline
2(C\mapsto B)&\binom{x+3}{2}&-\rho/2\\
(C\mapsto A)+(C\mapsto B)&(x+3)(x+2)&(1-2\rho)/3\\
2(C\mapsto A)&\binom{x+3}{2}&1-\rho\\
(B\mapsto A)+(C\mapsto A)&x(x+3)&(1-\rho)/2\\
(A\mapsto B)+(C\mapsto B)&(x-1)(x+3)&-\rho/2
\end{array}
\label{eq:app-E9}
\end{equation}
\endgroup
Their sums are positive multiples of $\kappa_{h-k}-\rho$ and $\kappa_k-\rho$, respectively. Here $x=2k=h-1$, and
\begin{equation}
 \kappa_{h-k}:=\frac{11x+16}{13x+20},
 \qquad
 \kappa_k:=\frac{16x+20}{29x+28}.
\label{eq:app-E10}
\end{equation}

Equations~\eqref{eq:app-E7} and~\eqref{eq:app-E10} define $\kappa_k$ for $1\leq k\leq s$, while~\eqref{eq:app-E5} and~\eqref{eq:app-E10} define $\kappa_{h-k}$ for $0\leq k\leq s$. These disjoint index ranges contain all $h$ switching values.

Replacing $x$ by $x+2$ also replaces $d$ by $d-2$ and $t$ by $t-3$. Substitution into~\eqref{eq:app-E7} gives, for $x=2k>0$ and $d\geq5$,
\begingroup\small
\[
\kappa_{k+1}-\kappa_k
 =\frac{4(9dx+3d+6x^2+1)}
 {(9d^2+12dx-14x-13)(9d^2+12dx+12d+10x+3)}>0.
\]
Similarly, substitution into~\eqref{eq:app-E5} gives, for $x=2k\geq0$ and $d\geq5$,
\begin{align*}
    \kappa_{h-k}&-\kappa_{h-k-1}\\
 &=\frac{4(9d^2+36dx+36d+24x^2+36x+11)}
 {(9d^2+12dx+6d-16x-27)(9d^2+12dx+18d+8x+5)}>0.
\end{align*}

\endgroup
The two denominator factors that are not visibly positive are $9d^2-13+(12d-14)x$ and $9d^2+6d-27+(12d-16)x$; both are positive for $d\geq5$. At the two $d=3$ boundaries, with $x>0$ for the first comparison, the corresponding differences are
\[
 \kappa_{k+1}-\kappa_k
 =\frac{20x^2+109x+110}{(23x+60)(29x+86)}>0,
\]
and
\[
 \kappa_{h-k}-\kappa_{h-k-1}
 =\frac{9x^2+47x+50}{(11x+35)(13x+46)}>0.
\]
The split thresholds lie in $(1/2,2/3)$. For~\eqref{eq:app-E7}, the coefficient sum at $\rho=1/2$ has sign $1+(t+1)/(x+t+1)>0$, whereas at $\rho=2/3$ it has sign $-t-1-(x-2)(t+1)/(x+t+1)<0$; here a split voter has even $x\geq2$. The pooled threshold in~\eqref{eq:app-E5} equals $1-1/(3+t+tx/(x+t+1))$ and lies in $(2/3,1)$. The two roots in~\eqref{eq:app-E10} have the same respective ranges because $x=h-1\geq6$. Thus $\kappa_s<2/3<\kappa_{s+1}$, completing the ordering.
Consequently,
\begin{equation}
 \frac12<\kappa_1<\kappa_2<\cdots<\kappa_h<1.
\label{eq:app-E11}
\end{equation}
We can now read off the update. The odd-$x$ calculation already showed that every pair changes configuration. If $x=2k$, a split pair remains split exactly when $\rho<\kappa_k$, and a pooled pair becomes split exactly when $\rho<\kappa_{h-k}$. Giving both members of pair $i$ utilities in $(\kappa_{\tau_i},\kappa_{\tau_i+1})$ therefore gives
\[
 \Phi(I)=P\setminus I
 \qquad(x\text{ odd}),
\]
and
\[
 \Phi(I)=\{i\in I:\tau_i<k\}
 \cup\{i\in P\setminus I:\tau_i<h-k\}
 \qquad(x=2k).
\]
All inequalities are strict. If $A$ is selected, then $U^p(A)>U^p(C)$ for all sufficiently small $p$, while $U^p(A)>U^p(B)$ by Lemma~\ref{lem:top-bottom-dominance}. If $C$ is selected, then $U^p(C)>U^p(A)>U^p(B)$. Exchanging $A$ and $B$ proves the corresponding statement for the other member of every pair. Thus paired profiles map to paired profiles.

Finally, consider an added voter with ranking $C\succ A\succ B$ and normalized utilities $(1,\gamma,0)$ in ranking order, where $0<\gamma<1/2$. Her private vector is $q^C$ for every $x$, including $x=h$. The same flow sums give her $C-A$ margin the signs $1-2\gamma$ for $d=0$, $1-\gamma$ for positive even $d$, and $(13x+20)-(11x+16)\gamma$ for $d=1$. For odd $d\geq3$, with $t=(3d-1)/2$, its sign is $\left(2+t+tx/(x+t+1)\right)(1-\gamma)+1$. All four expressions are positive. Lemma~\ref{lem:top-bottom-dominance} excludes $B$, so every added voter stays at $C$.
\end{proof}
\endgroup
\end{document}